\documentclass[11pt]{article}
\usepackage[usenames,dvipsnames,svgnames,table]{xcolor}
\usepackage{enumitem}
\usepackage{graphicx}
\usepackage{fancybox}
\usepackage{comment}
\usepackage{xcolor}
\usepackage{nameref}

\definecolor{ForestGreen}{rgb}{0.1333,0.5451,0.1333}
\definecolor{DarkRed}{rgb}{0.8,0,0}
\definecolor{Red}{rgb}{1,0,0}
\usepackage[linktocpage=true,
pagebackref=true,colorlinks,
linkcolor=ForestGreen,citecolor=ForestGreen,
bookmarks,bookmarksopen,bookmarksnumbered]
{hyperref}

\usepackage{amsmath}
\usepackage{amssymb}
\usepackage{amsthm}

\usepackage[ruled, linesnumbered]{algorithm2e}
\usepackage[noend]{algorithmic}
\usepackage{subfig}

\usepackage{thm-restate}

\usepackage{float}
\usepackage{cleveref}

\newtheorem{theorem}{Theorem}[section]

\newtheorem{lemma}[theorem]{Lemma}
\newtheorem{observation}[theorem]{Observation}

\newtheorem{invariant}[theorem]{Invariant}

\newtheorem{property}[theorem]{Property}
\newtheorem{definition}[theorem]{Definition}
\newtheorem{remark}[theorem]{Remark}

\newtheorem*{theorem*}{Theorem}
\newtheorem*{corollary*}{Corollary}
\newtheorem*{conjecture*}{Conjecture}
\newtheorem*{lemma*}{Lemma}
\newtheorem*{thm*}{Theorem}
\newtheorem*{prop*}{Proposition}
\newtheorem*{obs*}{Observation}
\newtheorem*{definition*}{Definition}

\newtheorem*{remark*}{Remark}
\newtheorem*{rec*}{Recommendation}

\newenvironment{fminipage}%
  {\begin{Sbox}\begin{minipage}}%
  {\end{minipage}\end{Sbox}\fbox{\TheSbox}}

\def\floor#1{\left\lfloor #1 \right\rfloor}
\def\ceil#1{\left\lceil #1 \right\rceil}

\def\union{\cup}

\DeclareMathOperator*{\argmin}{arg\,min}

\newcommand{\polylog}{\text{ polylog}}

\usepackage{fullpage}

\DeclareMathOperator{\Bef}{Bef}
\DeclareMathOperator{\suffix}{suffix}

\newcommand{\tHalf}{\Bef (t) }
\newcommand{\tHalfInText}{\Bef (t) }
\newcommand{\DSReal}{\mathcal{DS}}
\newcommand{\DSFront}{\widehat{\mathcal{DS}}}
\newcommand{\DSBack}{\widetilde{\mathcal{DS}}}
\newcommand{\SSSPAbbre}{SSSP }
\newcommand{\APSPAbbrev}{APSP }

\newcommand*\samethanks[1][\value{footnote}]{\footnotemark[#1]}

\title{Incremental SSSP for Sparse Digraphs Beyond the Hopset Barrier}
\author{Rasmus Kyng\thanks{The research leading to these results has received funding from grant no. 200021 204787 of the Swiss National Science Foundation.}, \\ ETH Zurich \\ kyng@inf.ethz.ch \\  \and Simon Meierhans\samethanks, \\ ETH Zurich \\ mesimon@inf.ethz.ch \and Maximilian Probst Gutenberg\samethanks, \\ ETH Zurich \\ maximilian.probst@inf.ethz.ch}
\begin{document}
\date{}
\maketitle
\pagenumbering{gobble}

\begin{abstract}
    Given a directed, weighted graph $G=(V,E)$ undergoing edge insertions, the \emph{incremental} single-source shortest paths (SSSP) problem asks for the maintenance of approximate distances from a dedicated source $s$ while optimizing the total time required to process the insertion sequence of $m$ edges.
    
    Recently, Gutenberg, Williams and Wein [\hyperlink{cite.PGVWW20}{STOC'20}] introduced a deterministic $\tilde{O}(n^2)$ algorithm for this problem, achieving near linear time for very dense graphs. For sparse graphs, Chechik and Zhang [\hyperlink{cite.CZ21}{SODA'21}] recently presented a deterministic $\tilde{O}(m^{5/3})$ algorithm, and an adaptive randomized algorithm with run-time $\tilde{O}(m\sqrt{n} + m^{7/5})$. This algorithm is remarkable for two reasons: 1) in very spare graphs it reaches the directed hopset barrier of $\tilde{\Omega}(n^{3/2})$ that applied to all previous approaches for partially-dynamic SSSP  [\hyperlink{cite.HKN14}{STOC'14}, \hyperlink{cite.PGWN20}{SODA'20}, \hyperlink{cite.BPGWL20}{FOCS'20}] \emph{and} 2) it does not resort to a directed hopset technique itself. 
    
    In this article we introduce \emph{propagation synchronization}, a new technique for controlling the error build-up on paths throughout batches of insertions. This leads us to a significant improvement of the approach in [\hyperlink{cite.CZ21}{SODA'21}] yielding a \emph{deterministic} $\tilde{O}(m^{3/2})$ algorithm for the problem. By a very careful combination of our new technique with the sampling approach from [\hyperlink{cite.CZ21}{SODA'21}],  we further obtain an adaptive randomized algorithm with total update time $\tilde{O}(m^{4/3})$. This is the first partially-dynamic SSSP algorithm in sparse graphs to bypass the notorious directed hopset barrier which is often seen as the fundamental challenge towards achieving truly near-linear time algorithms.
\end{abstract}
\clearpage
\pagenumbering{arabic}

\section{Introduction}

The single source shortest paths (SSSP) problem is among the first algorithmic problems taught in undergraduate courses. Its broad applicability has caused its numerous variants to be well studied. One such variant is the fully or partially dynamic \SSSPAbbre problem, where distances have to be maintained in a changing environment. 

A dynamic graph refers to an $initial$ graph  $G = (V, E)$ and a sequence of $insertions$ and $deletions$, where an $insertion$ causes an element to be added to the edge set $E$, and a $deletion$ removes an element from said edge set. In the \emph{partially dynamic} setting, these changes are limited to either just insertions or just deletions, and we call the graph $incremental$ or $decremental$ respectively.

Partially dynamic graph algorithms are interesting for numerous reasons. They often serve as a proxy for the harder fully dynamic setting to develop new techniques, and even more often, reductions from the fully dynamic setting to a partially dynamic setting exist. Additionally, partially dynamic \SSSPAbbre in particular is often a crucial building block for solving more complex dynamic problems, such as All-Pairs Shortest Paths (APSP) \cite{K99, RZ11}, maintaining the diameter \cite{AHR19, CG20} or maintaining a matching in a bipartite graph \cite{BHR19}. Finally, algorithms developed for partially dynamic graphs can be used as building blocks when constructing static graph algorithms, most prominently, a recent decremental SSSP algorithm has served as a key component in the first near-linear time algorithm for approximate, undirected min-cost flow \cite{bernstein2021deterministic}. 

In this article, we tackle the approximate \SSSPAbbre problem for sparse, weighted, directed and incremental graphs. Given a dedicated source $s$, we aim to maintain distance estimates $\hat{d}(s,v)$ for all $v \in V$ throughout a sequence of edge insertions, such that the approximation guarantee $d(s,v) \leq \hat{d}(s,v) < (1 + \epsilon)d(s,v)$ always holds. In the randomized setting, the edge insertions are assumed to be generated by an adaptive adversary, meaning that the adversary may adapt the sequence of insertions based on information about the state of our algorithm revealed by its answers to queries. This enables the usage of the developed algorithms as a black box subroutine. 

\subsection{Prior Work}
\label{subsec:priorWork}
We discuss prior work that is directly related to the incremental SSSP problem. For a more detailed overview of related work, we refer the reader to \Cref{subsec:relWork}. Below, we let $m$ denote the total number of edges, $n$ the number of vertices and $W$ the ratio between the largest and smallest edge weight. For $(1+\epsilon)$-approximate algorithms, we assume that $\epsilon > 0$ is constant to ease presentation.

\paragraph{Lower Bounds.} Solving the fully dynamic \SSSPAbbre problem in time $O(n^{3 - \epsilon})$, for a graph incurring $O(n)$ updates, would imply a solution to the \APSPAbbrev problem in truly subcubic time $O(n^{3 - \epsilon})$ in a trivial way. It is widely believed that this is impossible \cite{Wil15, WW18}. 

If exact distances have to be maintained, the advantage of the partially dynamic setting is believed to be limited to reducing the runtime of a combinatorial algorithm to $\Omega(mn^{1 - o(1)})$, since numerous popular conjectures, such as the APSP hypothesis, the BMM conjecture and the OMv conjecture, forbid the existence of a faster such algorithm \cite{RZ04,AW14,HKT15b}. Conditional on the $k$-Cycle hypothesis, this holds for any algorithm, not just combinatorial ones \cite{PGVWW20}. All known lower bound techniques carry over from the incremental to the decremental setting by playing the insertion sequence backwards. These lower bounds are matched by the ES-tree, which we introduce in the next paragraph. 

\paragraph{The ES-tree and its Variants.} Research on the partially dynamic \SSSPAbbre problem was initiated by the ES-tree data structure developed by Even and Shiloach \cite{SE81}, which was subsequently extended to weighted, directed graphs by Henzinger and King \cite{HK95}. For maximum edge weight $W$, it can maintain exact distances in $O(mnW)$ total update time. Edge rounding techniques improve said update time to $\tilde{O}(mn\log(W)/\epsilon)$, trading the exactness for a $(1 + \epsilon)$-approximation guarantee. The ES-tree remains an integral building block for faster approximation schemes.

\paragraph{Decremental SSSP beyond the ES-tree.} 
The $O(mn)$ barrier was first surpassed in the decremental setting. Henzinger, Krinninger and Nanongkai \cite{HKN14, HKN15} provided a randomized, oblivious $\tilde{O}(mn^{0.9} \log W)$ algorithm, which was improved on by Bernstein, Probst Gutenberg and Wulff-Nilsen \cite{PGWN20, BPGWL20} culminating in a randomized $\tilde{O}(\min\{n^2 \log^4 W, mn^{2/3} \log^3 W \})$ algorithm against an oblivious adversary.

\paragraph{The Directed Hopset Barrier.} All these improvements over the classic ES-tree for sparse directed graphs stem from using hop sets of hop $h$, which are a small collection of edges (typically we aim for set sizes $\tilde{O}(m)$ and are interested in the case where $m \sim n$) that when added to the graph ensure that there always is an approximate shortest path in the new graph that consists of at most $h$ edges. Lower bounds for $h$ currently achieve $\Omega(n^{1/11})$ \cite{Hesse03,HP18}. While hop sets with hop $O(\sqrt{n})$ and $\tilde{O}(n)$ edges are almost trivial to construct, hop sets with hop $n^{0.5 - \epsilon}$ for any constant $\epsilon > 0$ are not known to exist. In a recent breakthrough result, giving the first near-linear time algorithm to construct hop sets with hop $n^{0.5+o(1)}$, the authors explicitly state that their techniques are unlikely to achieve an improvement in the hop \cite{JLS19}. We refer to this lack of progress on directed hopsets as the \emph{directed hopset barrier}, and point out that all current approaches relying on hopsets require at least $\tilde{\Omega}(n^{3/2})$ total update time if the directed hopset barrier cannot be broken. This follows since these approaches rely on running an ES-tree on the union of the graph and the hopset where the running time of the ES-tree becomes $\tilde{\Theta}(nh)$ for very sparse graphs. We refer the reader to \Cref{subsec:hop_set_barrier} for an introductory explanation of the usage of an ES-tree with a hop set.

\paragraph{Incremental SSSP beyond the ES-tree.} While the algorithm of Henzinger, Krinninger and Nanongkai \cite{HKN14, HKN15} seems to translate over to incremental graphs, Probst Gutenberg, Vassilevska Williams and Wein \cite{PGVWW20} were the first to explicitly study the incremental setting, introducing a deterministic algorithm with total update time $\tilde{O}(n^2 \log W)$. This result only left room for polynomial improvements on sparse graphs, for which Chechik and Zhang \cite{CZ21} recently proposed a deterministic algorithm with total update time $\tilde{O}(m^{5/3} \log W)$ and a randomized version with total update time $\tilde{O}((mn^{1/2} + m^{7/5}) \log W)$ which works against an adaptive adversary.  

\paragraph{Partially-Dynamic Single-Source Reachability (SSR) and Strongly-Connected Components (SCC). } 

Bernstein, Probst, and Wulff-Nilsen \cite{BPW19} introduced the first near linear algorithms for decremental SSR and SCC. While incremental SSR can also be solved in near linear time with a link-cut tree, incremental SCC poses more of a challenge. The deterministic algorithms of Haeupler et al. \cite{HKM12} and Bender et al. \cite{BFGT16} achieve total update time $O(m^{3/2})$ for incremental SCC, and Bender et al. \cite{BFGT16} present another deterministic algorithm for dense graphs with total update time $\tilde{O}(n^2)$. Very recently, Bernstein, Dudeja and Pettie \cite{bernsteinNewSCC} improved the bound for sparse graphs to $\tilde{O}(m^{4/3})$ using randomization. 

\newpage

\subsection{Our Contribution}

We first give a simple deterministic algorithm.

\begin{restatable}{theoremmain}{thmdet}
We construct a deterministic algorithm for the $(1 + \epsilon)$-approximate \SSSPAbbre problem in incremental, directed and weighted graphs with total update time $\tilde{O}(m^{3/2} \log(W)/\epsilon)$. Distance queries are answered in time $O(1)$ and an approximate shortest path $\pi_{s,x}$ is reported in time $O(|\pi_{s,x}|)$.
\label{thm:det}
\end{restatable}

This algorithm matches the directed hop set barrier $\tilde{\Omega}(n^{3/2})$ for very sparse graphs, and should be compared to the previously best deterministic algorithm that achieves $\tilde{O}(m^{5/3} \log(W)/\epsilon)$ update time \cite{CZ21}. 

Furthermore, we present a significantly faster randomized algorithm. 

\begin{restatable}{theoremmain}{thmrand}
We construct a randomized algorithm for the $(1 + \epsilon)$-approximate \SSSPAbbre problem in incremental, directed and weighted graphs with total expected update time $\tilde{O}(m^{4/3} \log(W)/ \epsilon^2)$ that answers all distance queries correctly with high probability $1 - O(1/n)$ against an adaptive adversary. Distance queries are answered in time $O(1)$ and an approximate shortest path $\pi_{s,x}$ is reported in time $O(|\pi_{s,x}|)$.
\label{thm:rand}
\end{restatable}

For $m = O(n^{9/8 - \epsilon})$ this algorithm surpasses the directed hop set barrier. It is the first algorithm in the partially dynamic setting to do so. It dominates the previous state-of-the-art randomized result of Chechik and Zhang \cite{CZ21} in the entire sparsity range $m = O(n^{3/2})$. For denser graphs the $\tilde{O}(n^2 \log W)$ algorithm of Probst Gutenberg, Vassilevska Williams and Wein \cite{PGVWW20} remains dominant.

\paragraph{Improving over the Hopset Barrier.} 

Whereas the randomized algorithm of Chechik and Zhang \cite{CZ21} seemed to approach the best total update time $\tilde{\Omega}(n^{3/2})$ perceivable under the directed hopset barrier, our randomized algorithm breaks this barrier for sparse graphs. Since all current lower bound techniques from the decremental setting carry over to the incremental one, this could either be interpreted as a sign that it is possible to 1) improve hopsets and break the hop set barrier, 2) design techniques in the decremental setting that avoid an approach relying on hopsets altogether, and/or 3) that current conditional lower bound tools lack some key distinguishing property between the incremental and the decremental world.

\paragraph{Relation to Incremental SCC.} 

In conjunction with the $\tilde{O}(n^2)$ algorithm by Probst Gutenberg, Vassilevska Williams and Wein \cite{PGVWW20}, our results for incremental single source shortest paths exactly match the bounds for incremental SCC: $\tilde{O}(\min(m^{3/2}, n^2))$ in the deterministic case, and $\tilde{O}(\min(m^{4/3}, n^2))$ when using randomization against an adaptive adversary. Although this might be a coincidence, we wonder whether there is a deeper connection between these two problems. 

\section{Preliminaries}

We consider an incremental weighted and directed graph $G = (V,E)$, undergoing edge insertions, but never deletions. Each edge $e = (u,v)$ can have any integer weight $\omega(u,v)$ between $1$ and some maximum weight $W$. We use $n = |V|$ as shorthand for the number of vertices and $m$ for the total number of edges added to the graph. We assume throughout that the graph remains simple. This is w.l.o.g. since our run-times do not scale in the number of vertices so a multi-edge can be inserted by splitting it via a vertex. 

Distances $d(u, v)$ refer to the minimal cumulative weight $\sum_{i = 1}^{l - 1} \omega(v_i, v_{i+1})$ achieved by any path $u= v_1, ..., v_l = v$ connecting $u$ and $v$ and $\infty$ if there is no such path. For a path segment $\sigma = v_1, ..., v_l$, we let $d_\sigma(v_i, v_j)$ denote $\sum_{k = i}^{j - 1} \omega(v_k, v_{k+1})$ for $i < j$. We use the shorthand $\hat{d}(v)$ for distance estimates $\hat{d}(s,v)$.

Further, we define the suffix of a path segment $\sigma = v_1, ..., v_l$. 
\begin{definition}
The $\suffix(\sigma, v_i)$ of a path segment $\sigma = v_1, ..., v_l$ is the segment $v_i, ..., v_l$.
\label{def:suffix}
\end{definition}

Short paths are easy to maintain using standard ES-trees. We use the following fact that is standard in the literature. 

\begin{theorem}[ES-tree, see \cite{bernstein2011improved}]
Let the number of hops $h$ denote the number of edges on a path. There exists a deterministic algorithm for maintaining single source shortest paths in incremental weighted directed graphs to hop $h$ in time $O(m h \log(nW) / \epsilon)$, maintaining $\hat{d}(v)$ such that $d(s,v) \leq \hat{d}(v) < (1+\epsilon) d(s,v)$ for each $v$ if there exists a shortest path connecting $s$ and $v$ with at most $h$ hops, and $d(s,v) \leq \hat{d}(v)$ otherwise.
\label{thm:es}
\end{theorem}

Throughout the article, we often write $x^{t}$ to denote a variable $x$ at the point in time when the algorithm has finished processing the $t$-th insertion. We let $x^0$ denote the variable after preprocessing the initial input graph $G$, and $E^{end}$ is the final edge set after all insertions were processed. When we want to emphasize that a variable $x$ belongs to a data structure $\DSReal$, we indicate it with a subscript, writing $x_{\DSReal}$.

We use $\lg(n)$ as a short-hand for $\log_2(n)$.

\section{Algorithms Overview}

\paragraph{Lazy ES-tree.} The ES-tree data structure has been very influential in the design of dynamic shortest path data structures. It stores the distance from the dedicated source $s$ for each vertex, and whenever an edge $(u,v)$ is inserted it checks if vertex $v$ profits from using this edge, and if so it recursively explores all the outgoing edges of $v$. A standard technique to create faster algorithms is to relax this behaviour, and only recursively explore if the stored distance of $v$ has decreased by a significant amount over time, which we denote as $\epsilon \delta$. This reduces the \emph{propagation} of changes significantly and the resulting algorithm is called a \emph{lazy} ES-tree. Frequently, separate lazy ES-trees are maintained for different path length ranges $[\tau, 2\tau)$, enabling $\delta$ to be chosen depending on $\tau$. 

\paragraph{Partial Dijkstra.} In our algorithm, we additionally use a procedure $\textsc{PartialDijkstra}(V_{input}, \epsilon)$ that takes a set of vertices $V_{input}$ and then runs Dijkstra on the induced graph $G[V_{fixed}]$, with $V_{fixed}$ being $V_{input}$ initially, where we then add additional vertices $v$ to $V_{fixed}$ on-the-go if relaxing an edge with tail in $V_{fixed}$ would decrease the distance estimate of $v$ by at least $\epsilon\delta$. By carefully implementing this procedure, we ensure that no path in the final graph $G[V_{fixed}]$ accumulates error.

We invoke this procedure on special sets of vertices $V_{input}$ regularly to mitigate large error on path segments due to the Lazy ES-tree implementation. For example a path $\pi_{s,x}$ might have an additive error of $\epsilon\delta/2$ on each edge due to the Lazy ES-tree update scheme. But inserting the first three vertices on $\pi_{s,x}$ into $V_{input}$ and invoking $\textsc{PartialDijkstra}(V_{input}, \epsilon)$ effectively removes the entire error on the path $\pi_{s,x}$. Our algorithm chooses sensible sets $V_{input}$.

\paragraph{The Deterministic Algorithm of Chechik and Zhang \cite{CZ21}.} Let us briefly review the deterministic algorithm in \cite{CZ21}: the algorithm checks for each inserted edge $e =(u,v)$ whether vertex $v$ \emph{profits} by $\epsilon\delta$, i.e. $\hat{d}(v) \geq \hat{d}(u) + \omega(u,v) + \epsilon \delta$. If not, then the algorithm moves to the next insertion. Otherwise, it sets $\hat{d}(v) =  \hat{d}(u) + \omega(u,v)$ and calls $\textsc{PartialDijkstra}(\{v\}, \epsilon)$. This triggers the Dijkstra-like propagation discussed above until nobody profits by $\epsilon \delta$ anymore. Since \textsc{PartialDijkstra} rules out error in-between explored vertices, it is natural to hope that after $B$ calls to $\textsc{PartialDijkstra}(\{v\}, \epsilon)$ the error on any $s - x$ path might only increase by $B\epsilon \delta$. Then we could batch $B = \sqrt{m}$ insertions to a phase where we update as described before for error tolerance $\delta \approx d(s,x)/\sqrt{m}$ within each phase. Upon ending a phase, we run Dijkstra to recompute all exact distance estimates from scratch, which is called a rebuild. Observe that there are a total of $m/B = \sqrt{m}$ phases. If one distinguishes between path length ranges, this would yield an $\tilde{O}(m^{3/2})$ algorithm. Chechik and Zhang \cite{CZ21} show that the error accumulation is no more than $B^2 \epsilon \delta$, and balancing the parameters yields their deterministic $\tilde{O}(m^{5/3})$ algorithm. 

Unfortunately, this analysis is tight, and the reason for this behaviour is simple:
although the way we set up the algorithm ensures that the error on \emph{each existing path} scales linearly, the adversary can insert new edges that \emph{connect} multiple paths which then ramps up more error. In fact, after $B$ insertions, there might be many paths that have error $\sim B$ each, and then another $B$ insertions can connect these paths to achieve combined error $\sim B^2$ on a single path. We sketch such an unfavourable sequence of edge insertions in Figure \ref{fig:counterexample}.

\begin{figure}[ht]
    \centering
    \includegraphics[width=13.5cm]{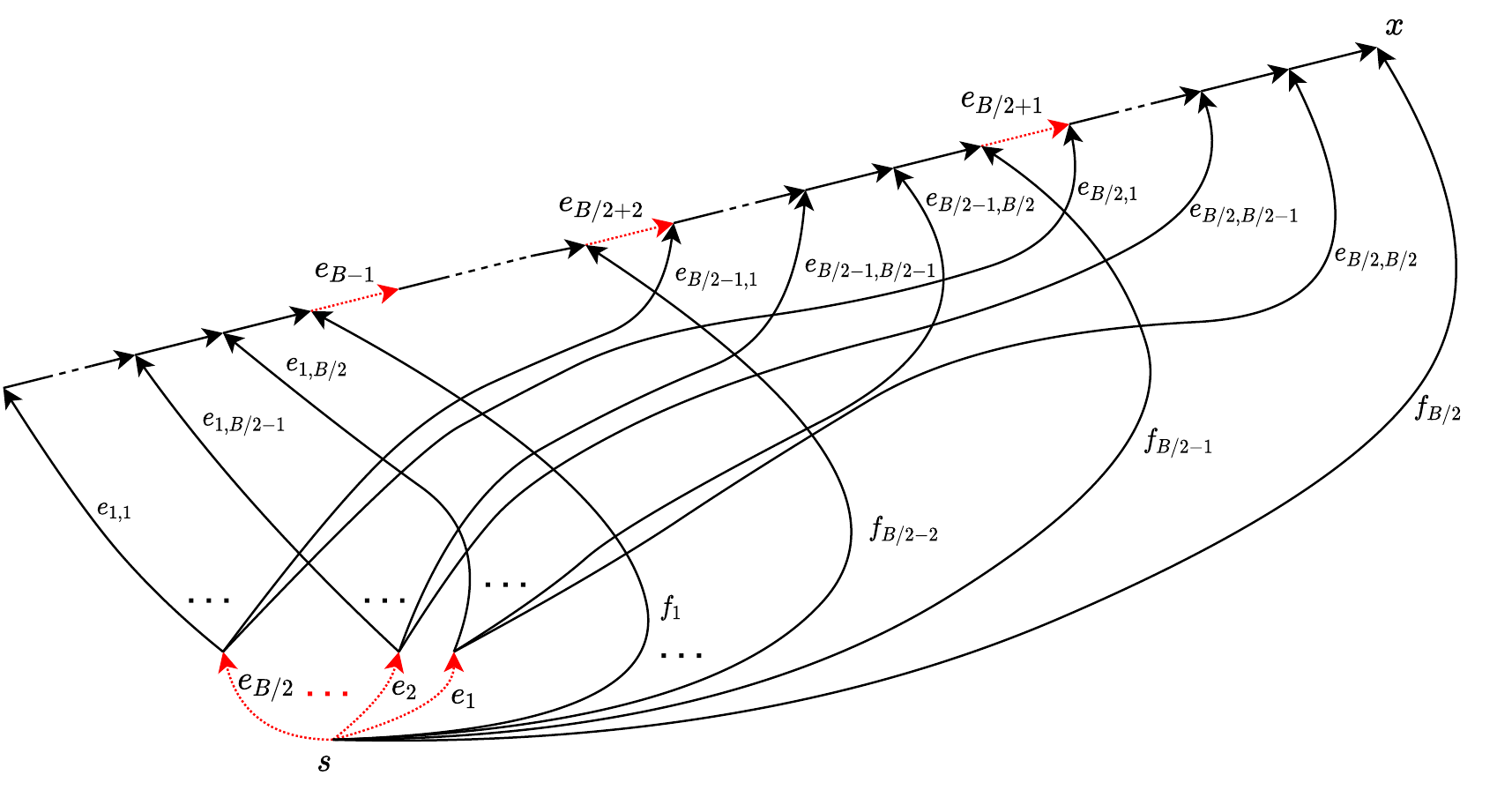}
    \caption{The dotted red edges $e_1, ..., e_{B - 1}$ are inserted consecutively during a single phase, the others were there before. The edges $f_i$ have weight $\omega(f_i) = i(1 + \frac{\epsilon \delta}{2})B/2 + i + 1$ and the edges $e_{i,j}$ have weight $\omega(e_{i,j}) = \omega(f_i) - (B/2 - j + 1)(1 + \frac{\epsilon \delta}{2}) - 1$. All the other edges, including the newly inserted red dotted ones, have weight one. All the vertices on the long path not incident to any edge $f_i$ start out having distance infinity. Then, the edge insertions cause the vertices on any segment to be lowered in distance estimate, one by one, in reverse order. Finally, the second half of the edge insertions connects the segments. Throughout all these insertions, $\hat{d}(x) = \omega(f_{B/2}) = B^2(1 + \frac{\epsilon \delta}{2})/4 + B/2 + 1 $ does not change, although there now is a path of length $B^2/4 + B/2 + 1$. This means that after roughly $m^{1/3}$ insertions, the additive error may reach $m^{2/3} \epsilon \delta$, making it necessary to exit the current phase and rebuild.}
    \label{fig:counterexample}
\end{figure}

\paragraph{Propagation synchronization.}
If we closely inspect the insertion pattern sketched in Figure \ref{fig:counterexample}, we note that the only way for error to build up in a path present since the last rebuild is if the distance estimate decreases happen in a very structured way: from the end to the start of the path segment, consecutively, and one by one. Had the insertions been done in any other order, the changes would have been propagated through. A very simple algorithm fueled by this insight would just add all vertices that were decreased in the current phase to $V_{input}$, which would rule out any build up of error on a segment present throughout the execution. But this operation is computationally expensive. However, by generalizing the pattern seen in the counterexample, we show that it suffices to add each such vertex to $\lg B$ sets $V_{input}$, drastically reducing the computational cost. 

Concretely, during the $i$-th insertion in a phase, we compute the maximum integer $j \in \mathbb{N}_0$ such that $i$ is divisible by $2^j$. We then add all the vertices that decreased significantly in distance estimate since time step $i - 2^j$ to $V_{input}$. On the one hand, this behaviour makes sure that a significant decrease causes a vertex to enter at most $\lg B$ sets $V_{input}$. On the other hand, after less than $B$ insertions happened in a phase, this scheme allows us to group the insertions into $\lg B$ batches, such that all the vertices touched by insertions of a given batch were in $V_{input}$ together at some point. At this moment, the propagation of these vertices are synchronized.

To illustrate this effect, we take a closer look at a segment $v_1, ..., v_{B/2 - 1}$ present throughout the phase from the worst case example sketched in Figure \ref{fig:counterexample}, taking $B$ a power of $2$ (see Figure \ref{fig:example_periodic}). The vertices $v_i$ get decreased one after another, from the end of the path segment to the front. But now, using our scheme, at time $B/4$, all the vertices decreased so far enter $V_{input}$, limiting the accumulation of error on their part of the segment to $\epsilon\delta$, because we propagate through vertices in $V_{input}$ no matter what. Iterating this argument limits the build up of error to $O(\epsilon \delta \log B)$. 

\begin{figure}[ht]
    \centering
    \includegraphics[width=15cm]{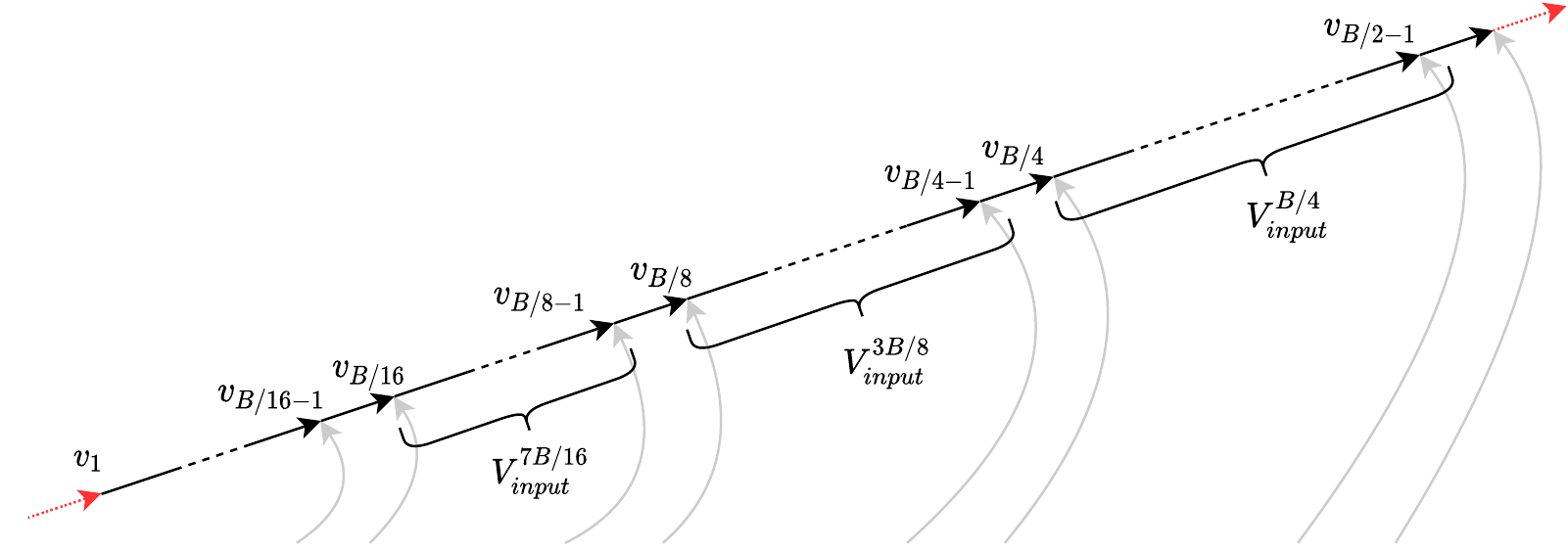}
    \caption{Zoom in on a segment present throughout a phase from the worst case example depicted in Figure \ref{fig:counterexample}. $V^{i}_{input}$ denotes the set $V_{input}$ passed to \textsc{PartialDijkstra} in step $i$ of the phase. Each batch of vertices within a curly bracket contributes at most $\epsilon\delta$ additive error, since all these vertices are in $V_{input}$ at some common moment. Therefore, this segment accumulates at most $O(\epsilon \delta \lg B)$ total additive error.}
    \label{fig:example_periodic}
\end{figure}

Generalizing this behaviour to arbitrary insertion patterns allows us to show that merely $2 \epsilon \delta$ error gets accumulated per batch, tallying up to $2 \epsilon \delta \lg B$. Thus, error is mainly contributed by connecting many different path segments where each segment makes only a small overall contribution. This technique essentially recovers the original intuition that error accumulation should be (near-)linear in the number of insertions. Moreover, the error on any high error path can be directly attributed to edges inserted to the path.

\begin{restatable*}{lemma}{lemfullpath}
Let $\pi_{s,x}$ be a shortest path from the source $s$ to $x$, after $b < B$ insertions happened during a phase starting at time $t$. Then $\hat{d}^{t+b}(x) \leq d^{t+b}(s,x) + 2 B\epsilon \delta \lg B + B \epsilon \delta$.
\label{lem:fullpath}
\end{restatable*}

We point out, however, that many details in proving \Cref{lem:fullpath} were swept under the rug in this exposition and it turns out that a very careful argument is required to prove the statement.

This result is the foundation of our deterministic algorithm, since it allows us to wait for considerably longer before we have to rebuild. In particular, after $B$ edges have been inserted, we incur an additive error of at most $3\epsilon \delta B\lg B$, whereas the additive error scaled with $B^2  \epsilon \delta$ before. As mentioned when introducing the ES-tree, we will build separate algorithms for distinct distance ranges. For a particular distance range $[\tau, 2\tau)$ the runtime is given by $\tilde{O}(m\tau/\epsilon\delta + m^2/B)$, where the first term is contributed by the operations performed throughout phases, and the second by the rebuilds. Setting $B \approx \sqrt{m}$ and $\delta \approx \tau / \sqrt{m}$ makes sure the error is no more than $\epsilon \delta B\ln B \approx \epsilon \tau$ whilst balancing the terms, and yields our deterministic algorithm.

\thmdet*

\paragraph{Randomization to Accelerate Rebuilding.} Fully rebuilding after $\sqrt{m}$ insertions is expensive, but seems necessary in the deterministic case. That is because, even when the adversary never inserts an edge with additive error more than $\epsilon \delta$, and thus \textsc{PartialDijkstra} is never called, there might be a high error path after $\sqrt{m}$ insertions. It is difficult to certify if these $\sqrt{m}$ insertions are part of the same shortest path.

To further speed-up our algorithm, we resort to randomization where we shorten phases to $B=m^{1/3}$ insertions, increase sensitivity to $\delta \approx \epsilon d(s,x)/m^{1/3}$ and aim to spend only $\tilde{O}(m^{2/3})$ time per rebuild. Note that this almost immediately implies a $\tilde{O}(m^{4/3})$ algorithm. 

\paragraph{A First Attempt at Randomization.} Again, we use a simple example to illustrate the main idea of the randomized rebuild. We start by considering the problem of certifying large error on some path. Consider therefore that we ran the first phase, and at the end of the phase, we have to identify for a vertex $x$ whether its shortest path has a total of at least $\epsilon d(s,x)$ additive error on the path. If so, the described technique is guaranteed to reduce $\hat{d}(x)$ by at least a bit.

We partition the interval $[0, \hat{d}(x)]$ into $O(m^{1/3})$ equal-sized sub-intervals $[8i\delta, 8(i+1)\delta)$ for all $i$, and then sample a single such interval $[8i\delta, 8(i+1)\delta)$ uniformly at random, and insert all vertices with distance estimate $\hat{d}(v) \in [8i\delta, 8(i+1)\delta)$ into a set $V^\star$ and invoke $\textsc{PartialDijkstra}(V^{\star}, \epsilon)$. Now assume the shortest path $\pi_{s,x}$ from $s$ to $x$ carries $\epsilon d(s,x)$ additive error. Since the error on any single edge is bounded by $\epsilon \delta$, the error must be well distributed among the distance ranges. This allows us to show that with constant probability, the error on the path segment on $\pi_{s,x}$ that crosses the sampled set $V^{\star}$ carries $\geq \epsilon\delta$ additive error. Therefore, this procedure causes a propagation that decreases the distance estimates up to  $\hat{d}(x)$ with good probability. Repeating this process $\Theta(\log n)$ times allows us to conclude that $\hat{d}(x)$ got decreased with high probability. 

\paragraph{A Path with Gaps.}

Now, consider a path $\pi_{s,x}$ from $s$ to $x$ that eventually becomes the shortest path but still has $m^{1/3}$ edges missing, waiting to be inserted in the upcoming phase to build a high error path. This partitions the path into $m^{1/3} + 1$ segments, and let us assume that each segment carries additive error $2\epsilon\delta$ so combined they carry $\sim 2\epsilon d(s,x)$ error.

The above sampling procedure still might sample a segment and would remove the error there but after $\Theta(\log n)$ samples we cannot correctly conclude that there is no path with huge error build-up. An adaptive adversary could now just observe which segments were sampled and increase the weight of the edge that connects the fixed segment to the upcoming one, only losing $2\epsilon \delta$ additive error on the whole path per fixed segment. To deal with this issue, Chechik and Zhang \cite{CZ21} fix \emph{many} segments, paying with a decrease of the potential $\Phi = \sum_{v \in V}  \hat{d}(v)$. This strategy is enabled by inserting fewer edges per phase, causing more error per segment and thus a more significant potential drop\footnote{They also shorten phases in case the potential gets decreased a lot. We omit this technical detail and we believe it could be removed.}. However, this remedy increases the runtime drastically.

\paragraph{Fix Now, Propagate Later.} Our techniques offer a solution to this issue. We run our algorithm against an oblivious adversary, and show how to recover adaptiveness later. Reconsidering the troublesome example where $m^{1/3} + 1$ segments contribute roughly $2\epsilon \delta$ error each, we discover that fixing a single segment might just be enough in this setting. Since the insertion sequence of an oblivious adversary is fixed, they have no way of adapting to which of the segments we fixed. However, if they do not adapt their sequence, the decrease of the fixed segment slowly propagates to the end of the path throughout the $m^{1/3}$ insertions. That is, since the edge connecting the fixed segment to the next carries more than $\epsilon \delta$ error when inserted, causing a propagation that fixes the next segment, and so on. See Figure \ref{fig:rand_fix} for an illustration of this crucial concept. 

Unfortunately, proving this propagation is quite involved, and requires new technical insights into the structure of the problem. In particular, we need a more careful argument about our synchronized propagation technique to bound the additive error building up on segments throughout the insertion of the new edges. Here, we omit these details and only focus on the central approach.

\begin{figure}[htp]
    \centering
    \includegraphics[width=15cm]{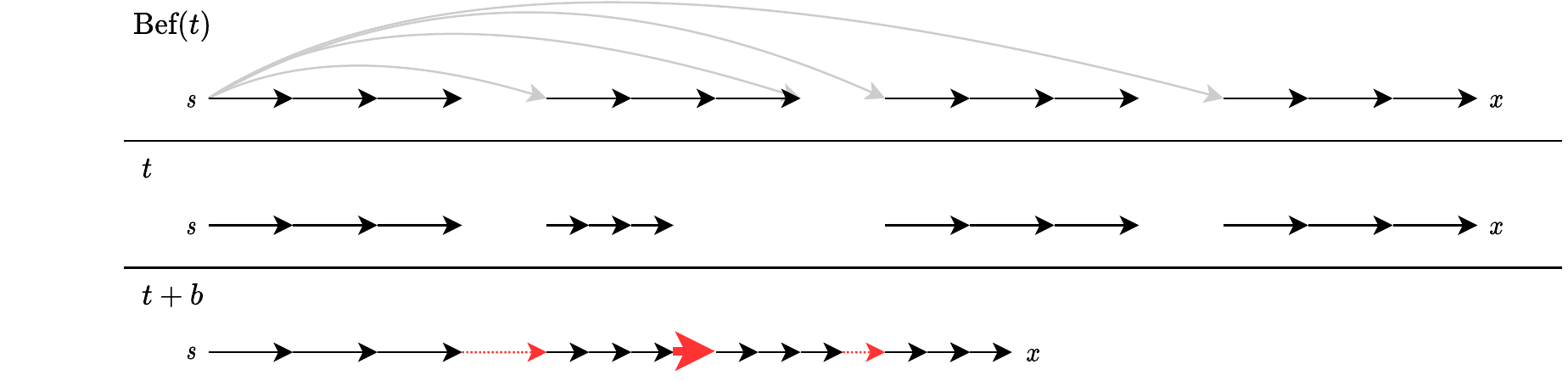}
    \caption{We distinguish between three important times: $t + b$ denotes the time for which we want to certify low error, $\tHalfInText$ denotes the time just before the last randomized rebuild before $t+b$, and $t$ the time right after said randomized rebuild. At time $\tHalfInText$, four segments are lined up to be connected to a high error $s-x$ path. During the randomized rebuild, the second segment is randomly chosen and fixed subsequently. But then, when the bold red edge is added between time $t$ and $t +  b$, it carries a lot of additive error because of our fix, causing propagations to ultimately remove the error from the remainder of the path. Notice that an adaptive adversary could just increase the weight of the bold red edge, foiling our fixing attempt.}
    \label{fig:rand_fix}
\end{figure}

\paragraph{Hiding the State from an Adaptive Adversary.}

To deal with an adaptive adversary, we hide which segment we fixed during the last rebuild. A simple technique that allows us to do so is to maintain two data structures $\DSReal$ and $\DSBack$, one from which we report distances, and an internal one whose state is hidden from the adversary. Whenever distances differ significantly between the two, we simply synchronize the two data structures $\DSReal$ and $\DSBack$ by setting all distance estimates to the minimum observed among the two data structures. Then, we trigger a rebuild on $\DSBack$ that remains hidden from the adversary. Such a rebuild only happens if the potential $\Phi_{\DSBack} = \sum_{v \in V} \hat{d}_{\DSBack}(v)$ has decreased drastically or many insertions happened, and therefore a very limited number of times. This combination of randomization and synchronized propagation yields our fastest algorithm. 

\thmrand*

\paragraph{Conclusion.} The newly introduced synchronized  propagation technique linearizes the build-up of error, enabling us to significantly improve on the state of the art for sparse digraphs. It also directly relates the additive error of a path to edge insertions to the same path, but it remains difficult to use this insight in the construction of an algorithm. 

\section{A Deterministic Algorithm}
\label{sec:detAlgorithm}

Our deterministic algorithm introduces the key concept of synchronized propagation, which we will also heavily make use of in our randomized scheme. 

\subsection{Algorithm Description}

As usual in the context of lazy ES-trees, we maintain $\tilde{O}(\log W)$ data structures for different distance ranges, where $W$ denotes the maximum weight on any edge. Namely, for each power of two $\tau = 2^i$ between $\sqrt{m}$ and $nW$, we maintain a data structure responsible for paths of length $[\tau, 2\tau)$. Each such data structure keeps distance estimates up to a maximum of $(1 + \epsilon)2\tau$, and returns infinity if a queried vertex is at this maximum distance. For short paths with length less than $2 \sqrt{m}$ we use a standard ES-tree (\Cref{thm:es}). Whenever a query has to be answered, the minimum distance among all data structures is returned.

We split the execution of our deterministic algorithm for maintaining such a distance range $[\tau, 2\tau)$ into multiple phases. After each phase, our distance estimates might have deteriorated, and we call Dijkstra with source $s$ to restore exact distance estimates. We call this a \emph{rebuild}. Each individual phase consists of $B = \floor{\sqrt{m}/(6\lg n)}$ edge insertions.

\begin{algorithm}
\caption{\textsc{insert}($u, v, \omega(u,v)$)}
\begin{algorithmic}[1]
\STATE Global variables $b$ and $last\_touched(v)$ get initialized to $0$.\\  Whenever a rebuild happens they are re-set to $0$ thereafter.
\STATE //\emph{insert edge}
\STATE $b \leftarrow b + 1$
\IF{$\ceil{\hat{d}(v)/\epsilon \delta} > \ceil{(\hat{d}(u) + \omega(u,v) )/\epsilon \delta}$}
    \STATE $\hat{d}(v) \leftarrow \hat{d}(u) + \omega(u,v)$
    \STATE $last\_touched(v) \leftarrow b$
\ENDIF
\STATE //\emph{synchronized propagation}
\STATE $j \leftarrow \max(j \in \mathbb{N}_0: \exists k \in
    \mathbb{N} \text{ s.t. } b = k \cdot 2^{j})$
\STATE $k \leftarrow b / 2^j $
\STATE $V_{input} \leftarrow \{v \in V: (k - 1) 2^j < last\_touched(v) \leq k 2^j = b\}$
\STATE $V_{touched} \leftarrow \textsc{PartialDijkstra}(V_{input}, \epsilon)$
\FORALL{$v \in V_{touched}$}
    \STATE $last\_touched(v) \leftarrow b$
\ENDFOR
\end{algorithmic}
\label{alg:det_insert}
\end{algorithm}

Let us next explain how we deal with the $i$-th edge insertion to our data structure (the corresponding pseudo-code is in \Cref{alg:det_insert}): when the $i$-th edge $e = (u,v)$ is inserted, we check if $\ceil{\hat{d}(v)/\epsilon \delta} > \ceil{(\hat{d}(u) + \omega(u,v) )/\epsilon \delta}$, where $\delta = \tau/\sqrt{m}$ is an error tolerance parameter depending on the path length range $\tau$. If so, we set  $\hat{d}(v) = \hat {d}(u) + \omega(u,v)$. Afterwards, if we added the $b$-th edge in this phase, we do the following. Let 
\begin{align*}
    j = \max(j \in \mathbb{N}_0: \exists k \in \mathbb{N} \text{ s.t. } b = k \cdot 2^{j})
\end{align*}
and let $k$ be the corresponding $k$ such that $b = k \cdot 2^{j}$. We measure the decrease in distance estimate of each vertex in steps of size $\epsilon \delta$, and we let $V_{input}$ denote the set of vertices that decreased to or past such an $\epsilon \delta$ step in distance estimate since the edge insertion $(k-1)2^j$ of the current phase was processed. Note that this includes all the vertices that accumulated a $\epsilon \delta$ decrease within this time frame. We call $\textsc{PartialDijkstra}(V_{input}, \epsilon)$ (\Cref{alg:partialdijkstra}) to propagate changes from these vertices, as well as along all edges that would end up with additive error $\epsilon \delta$ or more.

\begin{algorithm}
\caption{\textsc{PartialDijkstra}($V_{input}$, $\epsilon$)}
\begin{algorithmic}[1]
\STATE Initialize priority queue $Q \leftarrow V_{input}$ sorted by $\hat{d}(v)$ for $v \in V_{input}$.
\STATE $V_{touched} = \emptyset$
\WHILE{$Q \neq \emptyset$} 
\STATE extract min $u$ from $Q$
\FORALL{$v: (u,v) \in E$}
\IF{$\ceil{\hat{d}(v)/\epsilon \delta} > \ceil{(\hat{d}(u) + \omega(u,v) )/\epsilon \delta}$} \label{alg:line:touched}
    \STATE $\hat{d}(v) \leftarrow \hat{d}(u) + \omega(u,v)$
    \STATE $V_{touched} \leftarrow V_{touched} \cup \{v\}$
    \STATE add $v$ to $Q$ if $v \notin Q$
\ELSIF{$v \in Q$}
    \STATE $\hat{d}(v) \leftarrow \min(\hat{d}(v), \hat{d}(u) + \omega(u,v))$
\ENDIF
\ENDFOR
\ENDWHILE
\RETURN $V_{touched}$
\end{algorithmic}
\label{alg:partialdijkstra}
\end{algorithm}

\subsection{Runtime}

\begin{lemma}
A single decrease past an $\epsilon \delta$ step causes a vertex $v$ to enter $O(\log B)$ sets $V_{input}$.
\label{lem:part_dijksta_log}
\end{lemma}
\begin{proof}
Say the decrease happened at step $b$ since the start of a phase. To arrive at a contradiction, we assume this decrease caused $v$ to enter into $V_{input}$ at step $i_1 = k_1 2^j$ and $i_2 = k_2 2^j$ of the phase for $k_1 < k_2$. But then $(k_2 - 1) 2^j < b \leq k_1 2^j$ which is a contradiction since $k_1$ and $k_2$ are integers. The conclusion directly follows from the fact that the decrease causes $v$ to enter a set $V_{input}$ at most once for each $2^j$, and $2^j \leq B$.
\end{proof}

\begin{lemma}
The total runtime of the algorithm is $O(m^{3/2}\log^3(n)/\epsilon + m^{3/2}\log (W) \log^2(n)/\epsilon)$.
\label{lem:rt_det}
\end{lemma}
\begin{proof}
We first analyse a specific data structure for some distance range $\tau$. Two parts contribute to the runtime.

Firstly, we call Dijkstra on the whole graph after a batch of $B$ insertions, using $O(m^2 \log(n)/B)$ time. 

Secondly, for any edge $(u,v)$ we spend up to $O(\log B)$ time whenever $\hat{d}(u)$ gets decreased past an $\epsilon \delta$ step immediately and throughout subsequent calls to $\textsc{PartialDijkstra}(V_{input}, \epsilon)$ by \Cref{lem:part_dijksta_log}. While doing so, we spend at most another $\log n$ factor for sorting the priority queues. Since the total number of decreases past such steps is bounded by the initial potential $\Psi = \sum_{(u,v) \in E^{end}}  \hat{d}(u) = O(m\tau)$ divided by $\epsilon \delta$, such decreases can happen at most $O(m\tau/\epsilon\delta)$ times.

Therefore the total runtime is $O(\frac{m^2 \log(n)}{B} + \frac{m \tau \log^2(n)}{\epsilon \delta})$. Using the definitions $\delta = \tau/\sqrt{m}$ and $B = \floor{\sqrt{m}/(6\log n)}$ we get $O(\frac{m^{3/2} \log^2(n)}{\epsilon})$ per such data structure.

We maintain $O(\log n + \log W)$ data structures in total and keep the distances of paths shorter than $2\sqrt{m}$ in time $O(m^{3/2}/\epsilon)$ using \Cref{thm:es}. The claimed runtime follows. 
\end{proof}

\subsection{Correctness}

It is left to show that the described algorithm maintains the approximation guarantee. To do so, we focus on a path length in a specific range $[\tau, 2\tau)$, and show that the data structure responsible for said range does so. We first prove an important yet simple lemma about the $\textsc{PartialDijkstra}(V_{input}, \epsilon)$ routine. 

\begin{lemma}
Consider the graph right after calling $V_{touched} = $ \textsc{PartialDijkstra}($V_{input}, \epsilon$) for some vertex set $V_{input}$. Let $\sigma = v_1, ..., v_l$ be a path segment such that for all $i \in [l]: v_i \in V_{fixed}$ where $V_{fixed} = V_{touched} \cup  V_{input}$. Then 
\begin{align*}
    \hat{d}(v_l) \leq \hat{d}(v_1) + d_\sigma(v_1, v_l)
\end{align*}
where $d_\sigma(v_1, v_l)$ denotes the length of the path.
\label{lem:partialdijkstra}
\end{lemma}
\begin{proof}
For some $i \in [l - 1]$, consider the moment the vertex $v_i$ leaves the priority queue. Note that this is well defined, since each vertex $v_i \in V_{fixed}$ was added to the priority queue at some point, and therefore has to leave it for the algorithm to terminate. At that moment, either $v_{i+1}$ has already left beforehand, and thus $\hat{d}(v_{i+1}) \leq \hat{d}(v_{i})$ or we set $\hat{d}(v_{i+1}) = \min(\hat{d}(v_{i+1}), \hat{d}(v_{i}) + \omega(v_i, v_{i+1})$. The lemma follows directly by iterative applications of this inequality.  
\end{proof}

Next we show a simple and standard invariant, that often is useful in relaxed ES-trees. It summarizes the fact that an additive error of $\epsilon \delta$ on any edge causes a propagation, fixing said error. 

\begin{invariant}
If $(u,v) \in E$, then $\hat{d}(v) \leq \hat{d}(u) + \omega(u,v) + \epsilon \delta$.
\label{inv:max_str_edge}
\end{invariant}
\begin{proof}
Whenever an edge $(u,v)$ is added to the graph and $v$ decreases by $\epsilon \delta$ in distance estimate, we have $v \in V_{input}$ for the subsequent call to $\textsc{PartialDijkstra}(V_{input}, \epsilon)$. Thereafter, whenever a vertex distance estimate decreases by any amount, we check its out-neighbourhood and decrease the distance estimates of vertices in it if they violate the invariant as part of the $\textsc{PartialDijkstra}(V_{input}, \epsilon)$ routine. 
\end{proof}

With these basics out of the way, we focus on path segments that were already present in their current form at the start of a given phase, and aim to show that they can only accumulate a very limited amount of error. To do so, we define the slack of a segment, which is the maximum additive error any vertex on the segment witnesses with respect to the last vertex on the segment.

\begin{figure}[htp]
    \centering
    \includegraphics[width=5cm]{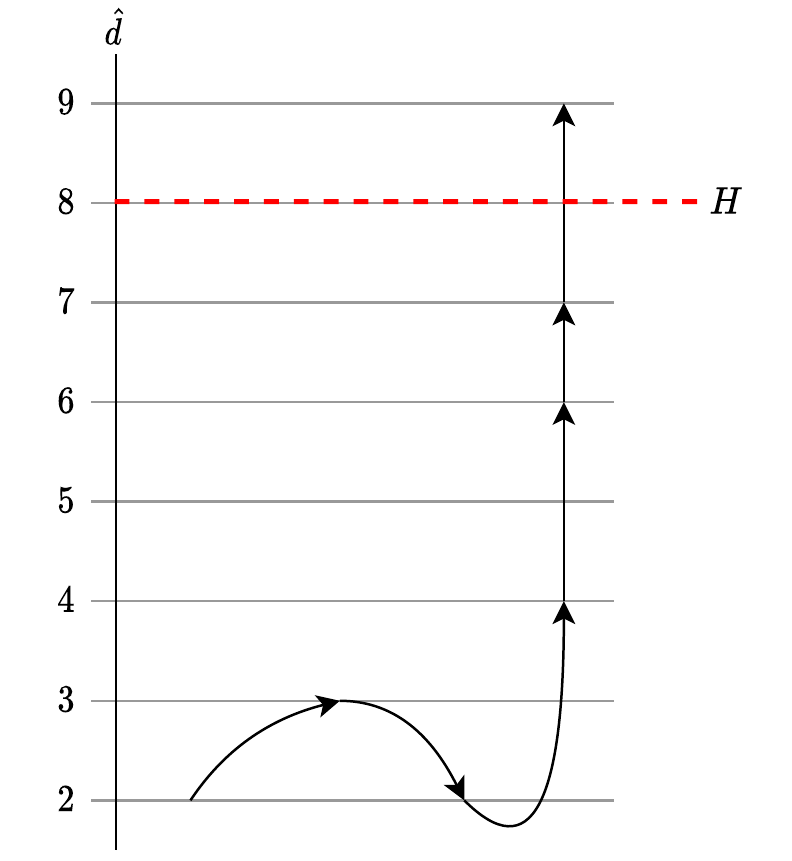}
    \caption{Illustration of \Cref{def:slack}. All displayed edges have weight $1$, but sometimes distance estimates increase by more than $1$ along the path, and sometimes they even decrease. The current slack of the segment is $slack(\sigma) = 3$, since the vertex witnessing the most error has distance $4$ to the last vertex, but the difference in distance estimates is $9 - 2 = 7$. If we measure slack with respect to some fixed height $H = 8$, the segment just has $slack(\sigma, H) = 2$. Slack with respect to fixed heights will only become important in the analysis of our randomized algorithm.}
    \label{fig:slack}
\end{figure}

\begin{definition}
Given a path segment $\sigma = v_1, ..., v_l$ of length $d_\sigma(v_1,v_l)$, we define its slack at time $t$ as
\begin{align*}
    slack^t(\sigma) := \max_{i \in [l]}(\hat{d}^t(v_l) - \hat{d}^t(v_i) - d_\sigma(v_i,v_l))
\end{align*}
and the slack with respect to some fixed height $H$ as
\begin{align*}
    slack^t(\sigma, H) := \max_{i \in [l]}(H - \hat{d}^t(v_i) - d_\sigma(v_i,v_l)).
\end{align*}
We sometimes omit the time if it is clear from the context. 
\label{def:slack}
\end{definition}

\begin{remark}
The definition of slack is with regard to the distance estimates maintained by the data structure. When we want to emphasize that we are using slack defined with regard to a data structure $\DSReal$, we use the variable name of the data structure as a subscript, for example we write $slack^t_{\DSReal}(\sigma)$.
\end{remark}

In the following, we want to show that the slack only accumulates slowly on paths that are already present in the graph at the start of a phase. For now we focus on a batch of insertions within a quite specific time interval, such that they get fixed by a a call to $\textsc{PartialDijkstra}(V_{input}, \epsilon)$ in the end. 

\begin{lemma}
Consider the $i$-th insertion during a phase that started at time $t$, along with the associated values $j$ and $k$ computed by our algorithm. If an existing path segment $\sigma = v_1, ..., v_l$ had $slack^{t+t_1}(\sigma) \leq \mu$ for $t_1 = (k - 1)2^j$ and $\mu \geq 0$, we have $slack^{t+t_2}(\sigma) \leq \mu + 2 \epsilon \delta$ for $t_2 = i = k 2^j$.
\label{lem:qual_dec}
\end{lemma}
\begin{proof}
Consider any vertex $v_p$ for $1 \leq p \leq l$ and let $d_{\sigma}(u,v)$ denote distances on the considered path. We aim to show 
\begin{align*}
        \hat{d}^{t+t_2}(v_l) \leq \hat{d}^{t+t_2}(v_p) + d_\sigma(v_p, v_l) + \mu + 2 \epsilon \delta.
\end{align*} 
We assume $\hat{d}^{t+t_1}(v_l) > \hat{d}^{t+t_2}(v_p) + d_\sigma(v_p, v_l) + \mu + \epsilon \delta$. Otherwise we are done, using $\hat{d}^{t+t_2}(v_l) \leq \hat{d}^{t+t_1}(v_l)$ since our algorithm only ever decreases distance estimates. Consider the set $V_{fixed} = V_{touched} \cup V_{input}$ for $V_{touched} = \textsc{PartialDijkstra}(V_{input}, \epsilon)$ as computed during the processing of the $i$-th edge insertion. By our assumption, we have $v_p \in V_{touched}$ and thus $v_p \in V_{fixed}$, since its distance estimate decreased by at least $\epsilon \delta$. Let $v_p, ..., v_q$ be a maximal subsegment starting at $v_p$ such that all vertices are contained in $V_{fixed}$. Either $v_q = v_l$ and thus all the vertices on this remainder of the path are in $V_{fixed}$, or $v_{q+1}$ is the first vertex not in $V_{fixed}$. If $v_q = v_l$, we have 
\begin{align*}
    \hat{d}^{t + i}(v_l) \leq \hat{d}^{t + i}(v_p) + d_\sigma(v_p, v_l) 
\end{align*}
by \Cref{lem:partialdijkstra} and are done. Otherwise 
we obtain
\begin{align*}
    \hat{d}^{t+t_2}(v_l) \leq \hat{d}^{t+t_1}(v_l) &\leq \hat{d}^{t+t_1}(v_{q+1}) + d_\sigma(v_{q+1}, v_l) + \mu \\
    &< \hat{d}^{t+t_2}(v_{q+1}) + d_\sigma(v_{q+1}, v_l) + \mu + \epsilon \delta\\
    &\leq \hat{d}^{t+t_2}(v_{q}) + d_\sigma(v_{q}, v_l) + \mu + 2 \epsilon \delta\\
    &\leq \hat{d}^{t+t_2}(v_{p}) + d_\sigma(v_{p}, v_l) + \mu + 2\epsilon \delta
\end{align*}
where the third inequality is due to $v_{q+1} \notin V_{touched}$, the fourth inequality is justified by \Cref{inv:max_str_edge} and the last inequality is justified by \Cref{lem:partialdijkstra}.
\end{proof}

We directly use the previous lemma to give a guarantee for an arbitrary number of insertions. 

\begin{lemma}
Consider a path segment $\sigma = v_1, ..., v_l$ that was present at the start of a given phase at time $t$. After $b \leq B$ insertions we have $slack^{t+b}(\sigma) \leq 2 \epsilon \delta \lg B$.
\label{lem:presentpath}
\end{lemma}
\begin{proof}
Let us segment the $b$ insertions into consecutive disjoint batches. We let $b_1 = 2^{\floor{\lg b}}$, and $b_i = 2^{\floor{\lg (b - \sum_{q < i} b_q)} }$, denote the number of insertions belonging to the $i$-th batch.
Then, for the last insertion of the $i$-th batch, our algorithm computes $j = \floor{\lg (b - \sum_{j < i} b_j)}$ and $k = 1 + (\sum_{q < i} b_q)/2^j$, yielding $(k - 1)2^j = \sum_{q < i} b_q$. 

At the start of a phase the path has $slack^t(\sigma) \leq 0$, and each batch increases the slack by at most $2\epsilon \delta$ by \Cref{lem:qual_dec}. Since each batch at least halves the number of remaining insertions, we conclude that we have $slack^{t + b}(\sigma) \leq  2 \epsilon \delta \lg B$ at the end.
\end{proof}

Now consider a shortest $s-x$ path $\pi_{s,x}$. Our strategy is to simply break it up into segments that were already present at the start of a phase, yielding a bound on the total stretch.

\lemfullpath

\begin{proof}
We segment the path into at most $B$ segments, which were already present at the start of the phase. Then we apply \Cref{inv:max_str_edge} at most $B - 1$ 
times to bound the newly inserted edges and apply \Cref{lem:presentpath} $B$ times to bound the segments already present at the start of the phase. Chaining the inequalities yields the result.  
\end{proof}

By summarizing the above, we achieve our deterministic result. 

\thmdet*
\begin{proof}
Follows by combining the correctness guarantee of \Cref{lem:fullpath} with the runtime guarantee of \Cref{lem:rt_det}. To answer queries in constant time, we store an extra variable per vertex maintaining its minimum observed distance estimate among all data structures. Storing the vertex that caused the last decrease for every vertex in every data structure enables us to maintain an approximate shortest path tree $T$, which allows path reporting in time $O(|\pi_{s,x}|)$.
\end{proof}

\section{A Randomised Algorithm}

We now give the technical details of our main result: a Monte-Carlo randomized algorithm that achieves expected total update time $\tilde{O}(m^{4/3}\log W/\epsilon^2)$.

\subsection{Algorithm Description}
For each power of two $\tau = 2^i$ between $m^{1/3}$ and $nW$, we maintain a data structure responsible for paths of length $[\tau, 2\tau)$. Shorter paths are handled by a standard ES-tree, as in our deterministic algorithm.

We describe the algorithm for maintaining a $(1+\epsilon)$ distance approximation for shortest paths $\pi_{s,x}$ of length $d(s,x) \in [\tau, 2\tau)$ for a dedicated source $s$. Our algorithm uses the parameters $\delta = \tau/m^{1/3}$ for controlling the error per edge and $B = \floor{m^{1/3}}$ denoting the maximum number of insertions per phase. For each vertex $v \in V$, we maintain a distance estimate $\hat{d}(v)$, which initially is set to $\hat{d}(v) = \min\{ d(s,v), \tau_{max}\}$ where we let $\tau_{max} = (2 + 200\lg n \epsilon)\tau + 1$ denote the maximum distance estimate. During queries $\tau_{max}$ is treated as infinity. Whenever a query has to be answered, the minimum distance estimate among all data structures is returned.

Our algorithm maintains two distinct lazy ES-trees $\DSReal$ and $\DSBack$ employing synchronized propagation. We refer to variables in $\DSReal$ with $x_{\DSReal}$ and to variables in $\DSBack$ with $x_{\DSBack}$ respectively. When queried, we answer with the distance estimates stored in $\DSReal$. Our randomized data structures only differ from our deterministic algorithm in the nature and frequency of rebuild phases, and we still use \textsc{insert} (\Cref{alg:det_insert}) to insert, this time for both data structures $\DSReal$ and $\DSBack$ separately for each edge. 

If either the last rebuild happened $B$ insertions ago, or the potential $\Phi_{\DSBack} = \sum_{v \in V}  \hat{d}_{\DSBack}(v)$ got decreased by $m^{1/3}\tau/4\epsilon$ since before the last rebuild, we enter another rebuild. In the context of our randomized algorithm, we also call rebuilds global fixing phases. Notice that global fixing phases might happen back to back, without insertions in-between, if the previous global fixing phase decreased the potential $\Phi_{\DSBack}$ enough. We describe a global fixing phase next.  

\paragraph{Global Fixing Phase.} First, we replace each distance estimate of a vertex $v$ in $\DSReal$ and $\DSBack$ by the minimum of the two. Then we repeat the following for $\ceil{2000 \log n/\epsilon}$ iterations: Uniformly sample some $i$ from $0, ..., \lceil2m^{1/3} + 200 \epsilon m^{1/3} \lg n - 8\rceil$, and let 
\begin{align*}
    V^\star_j = \{v \in V: \hat{d}^{t+b}_{\DSBack}(v) \in [i \delta, (i+8) \delta)\}
\end{align*}
where $j$ denotes the iteration count.

Finally, we call $\textsc{PartialDijkstra}(\bigcup_{j = 1}^{\ceil{2000 \log n/\epsilon}} V^\star_j, \epsilon)$ (\Cref{alg:partialdijkstra}) on $\DSBack$ to propagate on these vertices all at once in data structure $\DSBack$. Pseudo-code for the global fixing phase is given in \Cref{alg:glob_fix}. We state the following observation which is straight-forward from the algorithm.

\begin{observation}
Our randomized algorithm preserves \Cref{inv:max_str_edge}.
\end{observation}

\begin{algorithm}[ht]
\caption{\textsc{GlobalFixingPhase}}
\begin{algorithmic}[1]
\STATE //\emph{synchronization}
\FORALL{$v \in V$}
    \STATE $x = \min (\hat{d}_{\DSReal}(v), \hat{d}_{\DSBack}(v)) $
    \STATE $\hat{d}_{\DSReal}(v) \leftarrow x$; $\hat{d}_{\DSBack(v)} \leftarrow x$
\ENDFOR
\STATE //\emph{referred to as time} $\tHalfInText$ 
\STATE //\emph{fixing a segment in} $\DSBack$
\STATE $V^\star = \emptyset$
\FORALL{$j = 1, ..., \ceil{2000 \log n/\epsilon}$}
    \STATE sample $i$ uniformly from $0, ..., \ceil{2m^{1/3} + 200 \epsilon m^{1/3} \lg n - 8}$
    \STATE $V^\star_j = \{v \in V: i \delta \leq \hat{d}_{\DSBack}(v) \leq (i + 8)\delta \}$
    \STATE $V^\star \leftarrow V^\star \union V^\star_j$
\ENDFOR
\STATE call $\textsc{PartialDijkstra}(V^\star, \epsilon)$ on $\DSBack$
\end{algorithmic}
\label{alg:glob_fix}
\end{algorithm}

\subsection{Runtime}
\begin{lemma}
The combined runtime of all data structures is $\tilde{O}(m^{4/3}\log (W)/\epsilon^2)$ in expectation.
\label{lem:rt_rand}
\end{lemma}
\begin{proof}
As previously, we analyse the run time of a single data structure for some distance range $[\tau, 2\tau)$ first. Let us analyze the following parts:
\begin{itemize}
    \item \underline{The time spent in $\DSReal$ and $\DSBack$:} for any edge $(u,v)$ we spend up to $O(\log B)$ time whenever $\hat{d}(u)$ gets decreased by $\epsilon \delta$ immediately and throughout subsequent calls from \textsc{Insert} to $\textsc{PartialDijkstra}(V_{input}, \epsilon)$ by  \Cref{lem:part_dijksta_log}. In doing so, we spend at most another $O(\log n)$ factor for sorting the priority queues. There can be at most $O(m\tau / \epsilon \delta)$ such decreases, since the potential $\Psi = \sum_{(u,v) \in E^{end}} \hat{d}_{\DSReal}(u) + \sum_{(u,v) \in E^{end}} \hat{d}_{\DSBack}(u)$ is at most $O(m \tau)$ at the start, and decreases by $\epsilon \delta$ whenever such a decrease happens. For our choice of $\delta$, we can therefore bound the total update time required by such updates with $O(m^{4/3} \log^2(n)/\epsilon)$. 
    \item \underline{The time spent in a global fixing phase caused by the insertion counter reaching $B = \floor{ m^{1/3}}$:} We have that there are at most $O(m^{2/3})$ such fixing phases since we have $\leq m$ insertions. Each global fixing phase has an expected extra runtime, i.e. runtime not payed for by a potential decrease which is accounted for in the data structures, of $O( m^{2/3} \log^2 (n)/\epsilon)$. That is, since there are $\Omega(m^{1/3})$ distance estimate ranges that could be sampled, and thus the sum of the degrees of the vertices in one of them is $O(m^{2/3})$ in expectation. We sample $O(\log (n)/\epsilon)$ such regions. This yields a total expected extra runtime of $O(m^{2/3} \log^2(n)/\epsilon)$ where we spend another $O(\log n)$ factor for sorting priority queues. The other edges that get explored during a global fixing phases phase are paid for by a decrease of the potential $\Psi$, as argued above. Computing the minimum of the distance estimates between $\DSReal$ and $\DSBack$ can be easily amortized over previous changes to the variables, if we store a list of all the variables in $\DSReal$ and $\DSBack$ that changed since the last global fixing phase, and only perform this operation on them. 
    \item \underline{Global fixing phases caused by a potential decrease:} Since the potential $\Phi_{\DSBack} = \sum_{v \in V} \hat{d}_{\DSBack}(v)$ is in $O(n\tau) \subseteq O(m \tau)$ at the start, at most $O(m^{2/3}/\epsilon)$ decreases by $m^{1/3}\tau/4\epsilon$ can happen. As argued in the previous item, this yields an expected extra runtime of $O(m^{2/3} \log^2(n)/\epsilon^2)$. Thus a single data structure can be maintained in $O(m^{4/3}\log^2(n)/\epsilon^2)$ total update time.
\end{itemize}

We use $\tilde{O}(\log W)$ separate data structures for different path lengths and a single standard ES-tree for paths with length less than $2m^{1/3}$ by \Cref{thm:es}. The claimed runtime follows.
\end{proof}

\subsection{Correctness}

\paragraph{Setup of the Proof.}

Some insertions are followed by one or multiple global fixing phases, while others are not. Let $t_1, t_2, \dots, t_k$ be the time stages where the global fixing phase is run at least once, where we also consider the initialization phase a fixing phase, so $t_1 = 0$. Then, we analyze our data structure by looking at the time steps $[t_i, t_{i+1})$ for all $i$. The last batch ends with the last insertion to the graph, and for convenience $t_{k+1}$ is set equal to the number of insertions plus one. We first state the main lemma of our proof, denoting the data structure responsible for path lengths in $[\tau, 2\tau)$ as $\DSReal_{[\tau, 2\tau)}$ in this paragraph.

\begin{lemma}
For $i \in [k]$, time step $t_i + b < t_{i+1}$ where $b \in \mathbb{N}_{\geq 0}$, and vertex $x$ such that $2m^{1/3} \leq d^{t_i + b}(s,x) \in [\tau, 2\tau)$. The probability that
\begin{align}\label{eq:fundamentalError}
    \hat{d}_{\DSReal_{[\tau, 2\tau)}}^{t_i + b}(x) \geq  d^{t_i + b}(s,x) + 100 \epsilon \tau \lg n
\end{align}
is less than $1/n^{5}$ against an adaptive adversary.
\label{lem:total_error_prob_start}
\end{lemma}

From this, it is simple to conclude the correctness of our algorithm. 

\begin{lemma}
With high probability $1 - O(1/n)$ 
\begin{align*}
    \hat{d}_{\DSReal_{[\tau, 2\tau)}}(x) <  d(s,x) + 100 \epsilon \tau \lg n \leq (1 + 100\epsilon \lg n) d(s,x) 
\end{align*}
holds for all paths so that $d(s,x) \in [\tau, 2\tau)$ after fully processing any insertion against an adaptive adversary.
\label{lem:correctness}
\end{lemma}
\begin{proof}
There are at most $m \leq n^2$ insertions, and after each insertion $n - 1$ paths could have to much error. The result follows from \Cref{lem:total_error_prob_start} by union bound, since we maintain distances shorter than $2m^{1/3}$ via a standard ES-tree (\Cref{thm:es}).
\end{proof}

Combined with the runtime guarantee, we conclude our main theorem. 

\thmrand*
\begin{proof}
Follows from \Cref{lem:rt_rand} and \Cref{lem:correctness} after scaling $\epsilon$. To answer queries in constant time, we store an extra variable per vertex maintaining its minimum observed distance estimate among all data structures. Storing the vertex that caused the last decrease for every vertex in every data structure enables us to maintain an approximate shortest path tree $T$, which allows path reporting in time $O(|\pi_{s,x}|)$.
\end{proof}

\paragraph{Proof of \Cref{lem:total_error_prob_start}.} We denote throughout the rest of the proof the time $t_i$ specified in the lemma by $t$, and simply write $\DSReal$ instead of $\DSReal_{[\tau, 2\tau)}$ to prevent clutter. To make our case, we also need to refer to a special moment in time whilst processing the $t$-th iteration. Namely, we let $\tHalfInText$ denote the time before the random propagation on $\DSBack$ during the last global fixing phase whilst processing insertion $t$ happened. Before we give our proof in full generality, let us initially make the following two strong, simplifying assumptions: 
\begin{itemize}
    \item \underline{An Oblivious Adversary:} In particular, we want to use in the proof the key property of an oblivious adversary that the adversary has to fix the update sequence before the algorithm is initialized. 
    \item \underline{Distance Estimates are reasonable after time step $\tHalfInText$:} We assume for every $x \in V$ so that $d^{\tHalf}(s,x) \in [\tau, 2\tau)$, 
    \[
    \hat{d}_{\DSReal}^{\tHalf}(x) <  d^{\tHalf}(s,x) + 100 \epsilon \tau \lg n.
    \] 
\end{itemize}

Given these two assumptions, the crux to the rest of the analysis is to condition on the random bits evaluated at time $\tHalfInText$. We then define a new data structure $\DSFront$ \emph{for the purpose of the analysis only}. We define $\DSFront$ to have the state of $\DSReal$ at time $\tHalfInText$ and thereafter we simulate our deterministic algorithm from \Cref{sec:detAlgorithm} \emph{with no rebuilds}. Note that $\DSFront$ and $\DSReal$ are in identical states up until the next global fixing phase after time $t$ occurs, i.e. until time $t_{i+1}$. 

Next, observe that by conditioning on the randomness at time $\tHalfInText$, by the fact that $\DSFront$ is updated by a deterministic algorithm thereafter, and since we assume an oblivious adversary which has to fix its update sequence in advance, the states of $\DSFront$ until the rest of the algorithms are determined at this point in time.

This allows us to find a time step after $t$, say the time step $t+b$ such that we have
\begin{align*}
    \hat{d}_{\DSFront}^{t + b}(x) \geq  d^{t + b}(s,x) + 100 \epsilon \tau \lg n
\end{align*}
for a fixed vertex $x$ with $d^{t+b}(x) \in [\tau, 2\tau)$. If no such time step $t+b$ exists, then we can conclude that for the rest of the algorithm, distance estimates in $\DSFront$ are always within the correct bounds for vertex $x$. Since $\DSReal$ has the same distance estimates as $\DSFront$ until the first global fixing phase after time $t$, we have that the same is true for $\DSFront$ in the time frame $[t, t_{i+1})$, and we are done.

Let us therefore focus on the case where $t+b$ exists. For the rest of the proof, we fix $t+b$ to be the minimal time step where the inequality \eqref{eq:fundamentalError} is satisfied for vertex $x$. We also fix $\pi_{s,x}$ to be the shortest path from $s$ to $x$ at time $t+b$ in the graph $G$. 

We distinguish again by cases:
\begin{itemize}
    \item \underline{If $t+b \geq t+B$:} Since we always run a global fixing phase after $B$ time steps, and therefore $t+B \geq t_{i+1}$, the lemma follows immediately. 
    \item \underline{If $t+b < t+B$:} In this case, we use the following key lemma.
    
    \begin{lemma}\label{lma:keyLemmaPotentialReduction}
    Given $t + b < B$ exists, in the time frame $[\tHalf, t+b]$, the data structure $\DSBack$ decreased the potential $\Phi_{\DSBack}$ by at least $m^{1/3}\tau/4\epsilon$ with probability $\geq 1 - n^{-5}$.
    \end{lemma}
    
    Given this lemma, it is clear that the potential reduction must have triggered a new global rebuilding phase, the latest at stage $t+b$. Therefore, $t_{i+1} \leq t+b$ which concludes the proof. 
\end{itemize}

Here, the final step omitted so far is to prove \Cref{lma:keyLemmaPotentialReduction}. This requires an extremely careful analysis that we present in the rest of this section.

\paragraph{Removing the Simplifying Assumptions from the Proof.} Before we start our analysis of the potential reduction, let us briefly address the simplifying assumptions we made, and argue that we do not in fact need them:
\begin{itemize}
    \item \underline{An Oblivious Adversary:} To see that our proof also works against an adaptive adversary, it only remains to observe that we report distances and paths exclusively based on the information in the data structure $\DSReal$. Thus, the adversary cannot use the query output to guess the random bits chosen in the last global rebuilding phase that are used to run the data structure $\DSBack$. Once $\DSBack$ reports the large potential decrease, we immediately run a new global fixing phase that selects new (hidden) random bits for $\DSBack$ after revealing information based on the old random bits to the adversary.
    \item \underline{Distance Estimates are reasonable after time step $\tHalfInText$:} We implicitly used at this point that at the end of stage $t$, the distances are reasonable, so that in our proof, if we enter the last case (where $t+b$ exists with $t+b < t + B$), we have that $b > 0$. The reason this was required in our proof so far is that the next Global Fixing Phase was per definition, the earliest at the stage $t+1$ since after time step $\tHalfInText$, we did not change $\DSReal$ at stage $t$ anymore. 
    
    However, the proof that for every $x \in V$ so that $d^{\tHalf}(s,x) \in [\tau, 2\tau)$, 
    \[
    \hat{d}_{\DSReal}^{\tHalf}(x) <  d^{\tHalf}(s,x) + 100 \epsilon \tau \lg n
    \] 
    is analogous to the proof we just discussed. The analysis then shows that, with high probability, the potential $\Phi_{\DSBack}$ got decreased by at least $m^{1/3}\tau/4\epsilon$ between $\tHalfInText$ and $t$, and thus $\tHalfInText$ does not refer to a moment within the last global fixing phase while processing $t$, since another global fixing phase is triggered by such a decrease. This is a contradiction.
\end{itemize}

\paragraph{Overview of the proof of the potential reduction.} 

Our proof of \Cref{lma:keyLemmaPotentialReduction} can be logically divided into four steps:

\begin{enumerate}
    \item \emph{The slack of the path: }Given that the path $\pi_{s,x}$ has high additive error in data structure $\DSFront$ at time $t + b$, not all that error can be contributed by the newly inserted edges. Therefore, segments already present at time $t$ must contribute a significant cumulative amount of error.
    \item \emph{The slack of a segment: }Since our synchronized propagation technique severely limits the build up of error on segments present throughout $b < B$ insertions, some of these segments must have already carried slack at time $\tHalfInText$, when the data structures $\DSFront$ and $\DSBack$ were equivalent. 
    \item \emph{Tense vertices: }During the remainder of the global fixing phase, one such segment looses a large part of its slack in data structure $\DSBack$.
    \item \emph{Tense segments: }This fixed segment renders it impossible to build the path $\pi_{s,x}$ with high error in $\DSBack$. But then a lot of vertices moved closer to the source in $\DSBack$ since time $\tHalfInText$, causing a large decrease in the potential $\Phi_{\DSBack}$. Such a potential decrease causes another global fixing phase, yielding $t_i + b \geq t_{i + 1}$.
\end{enumerate}

\begin{figure}[ht]
    \centering
    \includegraphics[width=10cm]{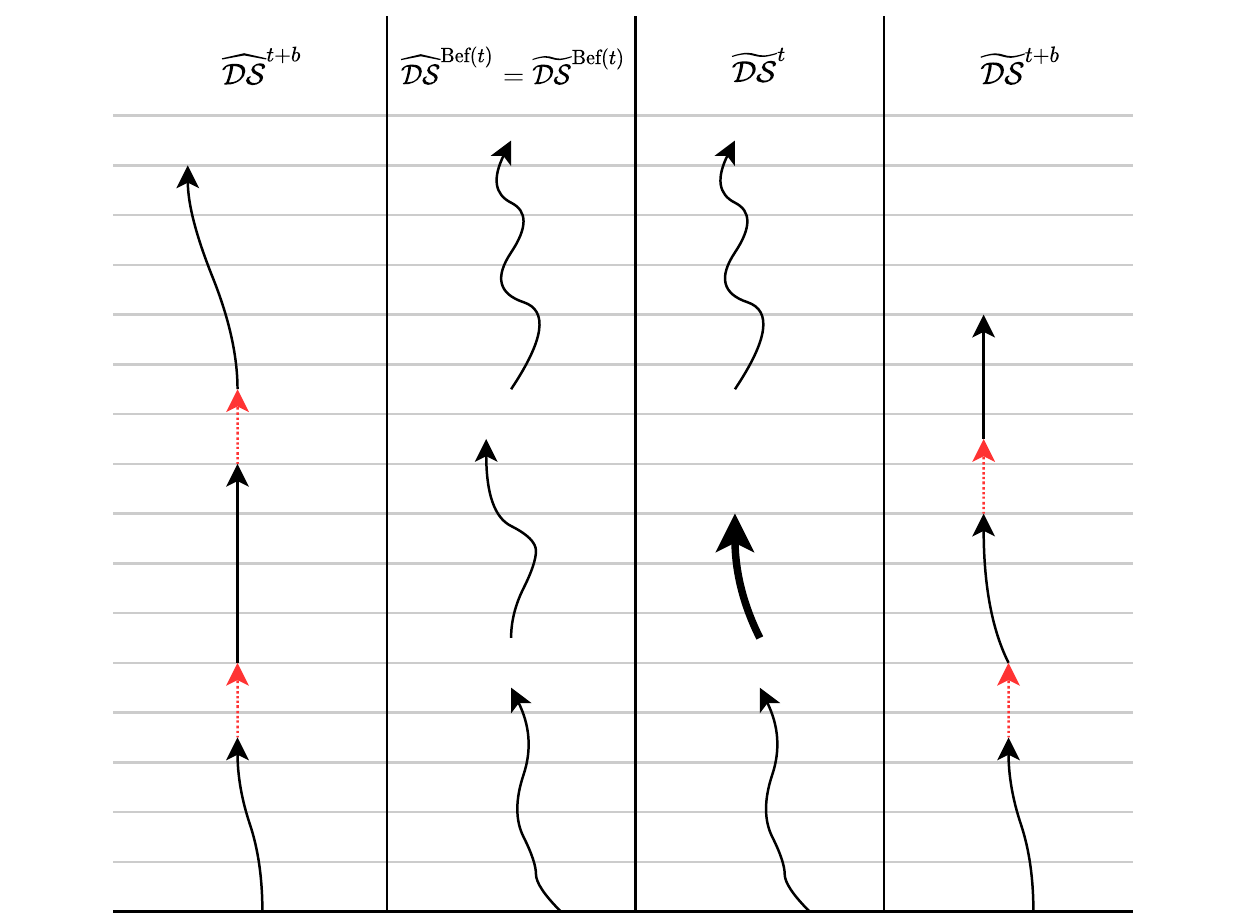}
    \caption{Given a high error path at time $t+b$ in data structure $\DSFront$, we show that the black segments already contributed error at time $\tHalfInText$, when $\DSFront$ and $\DSBack$ were last made identical. Then, our algorithm fixed one of these segments with high probability during the rest of the global fixing phase between $\tHalfInText$ and $t$ (depicted in bold). However, this only happens in $\DSBack$ while we do not implement this fix for $\DSFront$. In $\DSBack$, this leads to propagations eliminating most of the error on the remainder of the path when the dotted red edges are inserted. Finally, a path loosing this much error causes a drop in potential by at least $\epsilon m^{1/3}\tau/4$, and therefore $t_i + b \geq t_{i+1}$.}
    \label{fig:proof_overview}
\end{figure}

For an illustration of these steps, see Figure \ref{fig:proof_overview}. Due to some technicalities, it is simpler to present these steps in the following order: 2, 3, 4, 1. Since we want to show that the potential $\Phi_{\DSBack}$ gets decreased by at least $m^{1/3}\tau/4\epsilon$ in the time frame $[\{\tHalf\}, t+b]$, we assume no global fixing phases happened in this time frame. Otherwise we are done since $t + b < t + B$ and thus such a global fixing phase can only be caused by a $m^{1/3}\tau/4\epsilon$ decrease in potential. 

\paragraph{The slack of a segment.}

In this paragraph, we argue that a segment with large slack at time $t + b$ already carried some slack at time $t$, where we measure the slack with respect to some fixed distance estimate level $H$, tackling the challenge of segments moving around. Notice that a segment cannot move away from the source in distance estimate over time, since distance estimates only ever decrease. Therefore, when we want to argue about a segment's state at a previous moment, its vertices can only have had higher distance estimates. This means that finding an upper bound for the position at some previous moment suffices. First, we tighten and rephrase \Cref{lem:qual_dec}, making the accumulation of error precise. Note that we now state it from the perspective of going backwards in time: we presume something about a later state, and show a statement about a previous state. 

Since we assume no global fixing phase happened between $t$ and $t + b$, the data structures $\DSFront$ and $\DSBack$ are identical in their behaviour during this time interval, and the following arguments are equivalent for both. 

\begin{lemma}
Consider the $i$-th insertion during an insertion phase that started at time $t$, along with the associated values $j$ and $k$ computed by our algorithm. If a path segment $\sigma = v_1, ..., v_l$, present throughout the phase, has $slack_{\DSFront}^{t + t_2}(\pi, H) \geq \mu$ after insertion $t_2 = k 2^j = i$ of the phase for $\mu \geq 2\epsilon\delta$ and $H \leq \hat{d}_{\DSFront}^{t+t_2}(v_l)$, it had $slack_{\DSFront}^{t + t_1}(\pi,H) \geq \mu - 2\epsilon \delta$ after insertion $t_1 = (k - 1)2^j$ of the phase. This statement equivalently holds for data structure $\DSBack$.
\label{lem:qual_dec_tight}
\end{lemma}
\begin{proof}
Let $v_j$ be a vertex that witnesses the slack $\mu$ at time $t + t_2$, i.e. a vertex such that we have
\begin{align*}
    \mu \leq  H - \hat{d}_{\DSFront}^{t + t_2}(v_j) - d_\sigma(v_j,v_l)
\end{align*}
and in particular, using $H \leq \hat{d}_{\DSFront}^{t+t_2}(v_l)$, 
\begin{align*}
    \mu \leq \hat{d}_{\DSFront}^{t+t_2}(v_l) - \hat{d}_{\DSFront}^{t + t_2}(v_j) - d_\sigma(v_j,v_l).
\end{align*}
Consider the set $V_{fixed} = V_{touched} \cup V_{input}$ after the last call to \textsc{PartialDijkstra}($V_{input}$, $\epsilon$). Let $v_j, ..., v_q$ be a maximal sub-segment starting at $v_j$ such that all its vertices are elements of $V_{fixed}$. We first observe that $v_q  \neq v_l$, as otherwise 
\begin{align*}
    2\epsilon \delta \leq \mu \leq \hat{d}_{\DSFront}^{t+t_2}(v_l) - \hat{d}_{\DSFront}^{t + t_2}(v_j) - d_\sigma(v_j,v_l) \leq 0
\end{align*}
would hold by \Cref{lem:partialdijkstra} which is a contradiction. We calculate
\begin{align*}
    \mu &\leq H - \hat{d}_{\DSFront}^{t + t_2}(v_j) - d_\sigma(v_j,v_l) \\
        &\leq H - \hat{d}_{\DSFront}^{t + t_2}(v_q) - d_\sigma(v_q,v_l) \\
        &\leq H - \hat{d}_{\DSFront}^{t + t_2}(v_{q+1}) - d_\sigma(v_{q+1},v_l) + \epsilon \delta\\
        &\leq H - \hat{d}_{\DSFront}^{t + t_1}(v_{q+1}) - d_\sigma(v_{q+1},v_l) + 2\epsilon \delta
\end{align*}
where the second inequality follows from \Cref{lem:partialdijkstra}, the third by \Cref{inv:max_str_edge}, and the last by $v_{k+1} \notin V_{touched}$.
\end{proof}

Next we apply the previous lemma to a segment with high slack, and show that some of that slack was already present before the global fixing phase. Importantly, the slack is measured with respect to $H = \hat{d}_{\DSFront}^{t + b}(v_l)$ throughout, fixing its position in distance estimate. See  Figure \ref{fig:error_stays} for a depiction of our next lemma.

\begin{figure}[htp]
    \centering
    \includegraphics[width=5cm]{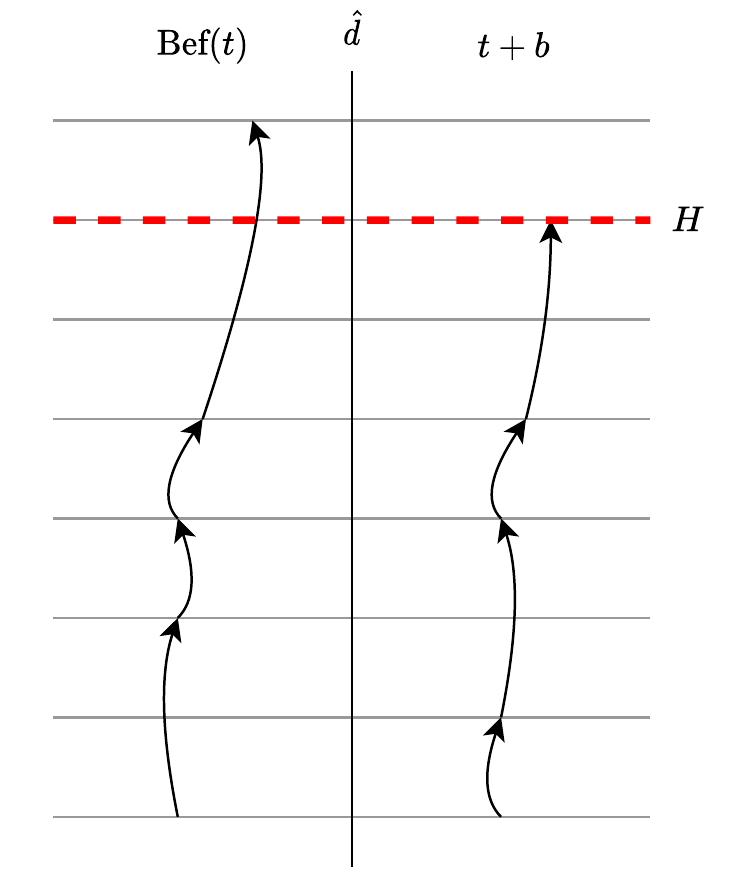}
    \caption{Illustration of \Cref{lem:error_stays} and \Cref{lem:error_stays_t_tilde}. A segment with large slack with respect to the distance estimate of the last vertex, denoted as $H$, had additive error in the same distance range at time $\tHalfInText$. It is crucial for our analysis that the segment did not only have high slack at time $\tHalfInText$, but this slack was built up by additive error below the distance estimate level $H$.}
    \label{fig:error_stays}
\end{figure}

\begin{lemma}
Consider a path segment $\sigma = v_1, ..., v_l$ already present during the last global fixing phase at time $t$, such that after $b$ insertions in the current phase $slack_{\DSFront}^{t+b}(\sigma) \geq \mu$ for $\mu \geq 3 \epsilon \delta \lg n$. We have $slack_{\DSFront}^{t}(\sigma, \hat{d}_{\DSFront}^{t + b}(v_l)) \geq \mu - 2  \epsilon \delta \lg n$. This statement equivalently holds for data structure $\DSBack$.
\label{lem:error_stays}
\end{lemma}
\begin{proof}
Let us segment the $b$ insertions into consecutive, disjoint batches. We let $b_1 = 2^{\floor{\lg b}}$, and $b_i = 2^{\floor{\lg (b - \sum_{j < i} b_j)}}$, denote the number of insertions belonging to the $i$-th batch. Clearly, there are at most $\lg B \leq \lg n$ batches. We work our way from the last batch to the first, and iteratively use \Cref{lem:qual_dec_tight}. We obtain that, after the $i$-th batch, we must have had
\begin{align*}
    slack_{\DSFront}^{t + t_i}(\sigma, \hat{d}_{\DSFront}^{t + b}(v_l)) \geq \mu - 2(p - i) \epsilon \delta
\end{align*}
for $t_i = b - \sum_{j = 1}^i b_i$. The lemma follows from $b \leq m^{1/3}$.
\end{proof}

We conclude this paragraph by giving a version of the previous lemma that only applies to $\DSBack$, and relates slack at time $t + b$ to slack at time $\tHalfInText$. 

\begin{lemma}
Consider a path segment $\sigma = v_1, ..., v_l$ already present during the last global fixing phase at time $t$, such that after $b$ insertions in the current phase $slack^{t+b}_{\DSBack}(\sigma) \geq \mu$ for $\mu \geq 3 \epsilon \delta \lg n$. We have $slack^{\tHalf}_{\DSBack}(\sigma, \hat{d}_{\DSBack}^{t + b}(v_l)) \geq \mu - 3  \epsilon \delta \lg n$.
\label{lem:error_stays_t_tilde}
\end{lemma}
\begin{proof}
By \Cref{lem:error_stays} we have
\begin{align*}
    slack^{t}_{\DSBack}(\sigma, \hat{d}_{\DSBack}^{t + b}(v_l)) \geq \mu - 2  \epsilon \delta \lg n.
\end{align*}
Since the global fixing phase between $\tHalfInText$ and $t$ just consists of a single call to \textsc{PartialDijkstra}, it adds at most $\epsilon\delta$ error. This follows directly from  \Cref{lem:partialdijkstra} and \Cref{inv:max_str_edge}, using the same argument as in the proof of \Cref{lem:qual_dec_tight}. Therefore, we have 
\begin{align*}
    slack^{\tHalf}_{\DSBack}(\sigma, \hat{d}^{t + b}(v_l)) \geq \mu - 2  \epsilon \delta \lg n - \epsilon \delta.
\end{align*}
\end{proof}

\paragraph{Tense vertices.} In this paragraph we argue that segments with high slack in $\DSBack$ are likely to be (partially) fixed by a global fixing phase, and the more slack they have, the more likely it gets.

We introduce the crucial concept of a $tense$ vertex. Informally speaking, a tense vertex is a vertex where a lot of error gets accumulated in a small forward neighbourhood on the path, and fixing this error causes the rest of the path to snap back by inducing a propagation all the way to the end of the path segment.

\begin{definition}[Tense Vertex]
For a segment $\sigma = v_1, ..., v_l$, a vertex $v_j$ is called \emph{tense}, if for all $v \in \suffix(\sigma, v_j)$, such that $\hat{d}^{\tHalf}_{\DSBack}(v) \geq \hat{d}^{\tHalf}_{\DSBack}(v_j) + 2\delta$, we have 
\begin{align*}
    \hat{d}^{\tHalf}_{\DSBack}(v) \geq \hat{d}^{\tHalf}_{\DSBack}(v_j) + d_{\sigma}(v_j, v) + \epsilon \delta.
\end{align*}
A tense vertex $v_j$ is \emph{hit}, if 
\begin{align*}
    \{v \in V: \hat{d}^{\tHalf}_{\DSBack}(v) \in [\hat{d}^{\tHalf}_{\DSBack}(v_j), \hat{d}^{\tHalf}_{\DSBack}(v_j) + 2\delta] \} \subseteq \bigcup_{j = 1}^{2000 \log n/\epsilon} V^\star_j.
\end{align*}
where the set $\bigcup_{j = 1}^{2000 \log n/\epsilon} V^\star_j$ is sampled during the last global fixing phase between time $\tHalfInText$ and $t$.
\label{def:tense_vertex}
\end{definition}

We first show that hitting a tense vertex eliminates the slack on the remainder of the segment. 

\begin{lemma}
Let $v_j$ be a tense vertex on a segment $\sigma = v_1, ..., v_l$ at moment $\tHalfInText$. If it gets hit between $\tHalfInText$ and $t$, we have $slack^{t}_{\DSBack}(\sigma') \leq 0$ for $\sigma' = \suffix(\sigma, v_j)$.
\label{lem:hit}
\end{lemma}
\begin{proof}
All vertices in $\sigma'$ that are not already in $V_{input}$ would profit by at least $\epsilon \delta$ by the definition of a tense vertex. Therefore the branching condition in Line \ref{alg:line:touched} in \Cref{alg:partialdijkstra} causes $v \in V_{touched} \cup V_{input}$ for each $v \in \sigma'$ given $V_{touched} = \textsc{PartialDijkstra}(V_{input}, \epsilon)$. The claimed error bound follows from \Cref{lem:partialdijkstra}. 
\end{proof}

Bearing this in mind, we show that the number of tense vertices on a segment with large slack scales linearly with said slack. We first define the collection of potentially tense vertices for a segment with respect to a parameter $r$. We introduce $r$ to measure how much the path gets pulled back, if one of them was tense and got hit. Later, we show that under favourable conditions, many of these vertices are actually tense. 

\begin{definition}[Potentially Tense Vertex Collection]
Consider a parameter $r \in \mathbb{R}_{\geq 0}$ and a segment $\sigma = v_1, ..., v_l$ already present during the last global fixing phase at time $t$ with
\begin{align*}
    slack^{\tHalf}_{\DSBack}(\sigma, \hat{d}^{t + b}_{\DSFront}(v_l)) - r = \mu > 0.
\end{align*}
Let $q = \floor{\mu/\epsilon \delta}$. We construct a collection of potentially tense vertices $w'_0, ..., w'_{\floor{q/2} - 1}$, where $w'_j$ is defined to be the last vertex along the path segment that has additive error at least $(q - 2j)\epsilon \delta + r$ with respect to $\hat{d}_{\DSFront}^{t + b}(v_l)$, i.e. the last vertex by regular order for which 
\begin{align*}
    \hat{d}_{\DSFront}^{t + b}(v_l) \geq \hat{d}^{\tHalf}_{\DSBack}(w_j') + d_\sigma(w_j', v_l) + (q - 2j)\epsilon \delta + r \tag*{(potentially tense)} \label{eq:potentially_tense}
\end{align*}
holds. These are the \emph{potentially tense} vertices of the segment $\sigma$ with respect to parameter $r$ at time $t + b$. 
\label{def:pot_tense_vert}
\end{definition}

We first show some basic properties of these vertices. We start by stating the following property which follows straight-forwardly from the definition of potentially tense vertices as it gets weakened as $j$ gets larger.

\begin{property}
Let $w_j'$ and $w_k'$ be potentially tense vertices of segment $\sigma = v_1, ..., v_l$ with respect to $r$ at time $t + b$ for $j < k$. Then either $w_j' = w_k'$ or $w_j'$ is before $w_k'$ on the path.
\end{property}

Further, by re-arranging the formula in the definition of potentially tense vertices we obtain the following property. 

\begin{property}
Let $w_j'$ be a potentially tense vertex of segment $\sigma = v_1, ..., v_l$ with respect to $r$ at time $t + b$. Then 
$\hat{d}^{\tHalf}_{\DSBack}(w_j') \leq \hat{d}_{\DSFront}^{t + b}(v_l) -  d_\sigma(w_j', v_l) - (q - 2j)\epsilon \delta - r $.
\label{prop:useful_ineq_tense_vertices_first}
\end{property}

\begin{lemma}
Let $w_j'$ be a potentially tense vertex of segment $\sigma = v_1, ..., v_l$ with respect to $r$ at time $t + b$. Then $\hat{d}^{\tHalf}_{\DSBack}(w_j') \geq \hat{d}_{\DSFront}^{t + b}(v_l) -  d_\sigma(w_j', v_l) - (q - 2j)\epsilon \delta - r - \epsilon \delta$.
\label{lem:useful_ineq_tense_vertices}
\end{lemma}
\begin{proof}
For the sake of contradiction, let us assume \begin{align}\label{eq:wjnottight}
        \hat{d}_{\DSFront}^{t + b}(v_l) > \hat{d}^{\tHalf}_{\DSBack}(w_j') + d_\sigma(w_j', v_l) + (q - 2j)\epsilon \delta + r + \epsilon \delta.
\end{align}
Since $w_j'$ is the last vertex for which $\hat{d}_{\DSFront}^{t + b}(v_l) \geq \hat{d}^{\tHalf}_{\DSBack}(w_j') + d_\sigma(w_j', v_l) + (q - 2j)\epsilon \delta + r$ holds, we must have by definition that for the successor $w_{next}$ of $w'_j$ on the segment $\sigma$
\begin{align}\label{eq:next_vertex}
    \hat{d}_{\DSFront}^{t + b}(v_l) < \hat{d}^{\tHalf}_{\DSBack}(w_{next}) + d_\sigma(w_{next}, v_l) + (q - 2j)\epsilon \delta + r.
\end{align}
Combining inequalities \eqref{eq:wjnottight} and \eqref{eq:next_vertex}, we obtain
\begin{align*}
    \hat{d}^{\tHalf}_{\DSBack}(w_{next}) - \hat{d}^{\tHalf}_{\DSBack}(w_j') > \epsilon \delta + \omega(w_j', w_{next})
\end{align*}
by moving and cancelling terms, which is a contradiction to  \Cref{inv:max_str_edge}.
\end{proof}

\begin{lemma}
Let $w_j'$ and $w_k'$ be potentially tense vertices of segment $\sigma = v_1, ..., v_l$ with respect to $r$ at time $t + b$ for $j < k$. Then $\hat{d}^{\tHalf}_{\DSBack}(w_k) - \hat{d}^{\tHalf}_{\DSBack}(w_j) \geq \epsilon \delta$ and thus they are distinct.
\label{lem:tens_vert_dist}
\end{lemma}
\begin{proof}
We use \Cref{prop:useful_ineq_tense_vertices_first} with $w_j'$ to derive $ \hat{d}^{\tHalf}_{\DSBack}(w_j') \leq \hat{d}_{\DSFront}^{t + b}(v_l) -  d_\sigma(w_j', v_l) - (q - 2j)\epsilon \delta - r$
 and \Cref{lem:useful_ineq_tense_vertices} with $w_k'$ to obtain $\hat{d}^{\tHalf}_{\DSBack}(w_k') \geq \hat{d}_{\DSFront}^{t + b}(v_l) -  d_\sigma(w_k', v_l) - (q - 2k)\epsilon \delta - r - \epsilon \delta$.
Subtracting the former from the latter inequality gives
\begin{align*}
    \hat{d}^{\tHalf}_{\DSBack}(w_k') - \hat{d}^{\tHalf}_{\DSBack}(w_j') \geq d_\sigma(w_j', w_k') + 2(k - j) \epsilon \delta - \epsilon \delta \geq \epsilon \delta.
\end{align*}
\end{proof}

\begin{lemma}
Let $w_j'$ be a potentially tense vertex of segment $\sigma = v_1, ..., v_l$ with respect to $r$ at time $t + b$ with $j < {\floor{q/2} - 1}$, so that $w_j'$ is not tense. Then $d_\sigma(w_j', w_{j+1}') \geq \delta$.
\label{lem:not_tense_big_dist}
\end{lemma}
\begin{proof}
We denote the next vertex along the path after $w_{j+1}'$ as $w_{next}$. Note that such a vertex always exists, since 
\begin{align*}
    \hat{d}_{\DSBack}^{\tHalf}(v_l)  = \hat{d}_{\DSFront}^{\tHalf}(v_l) \geq \hat{d}_{\DSFront}^{t + b}(v_l) > \hat{d}^{\tHalf}_{\DSBack}(w_j')
\end{align*} by the description of our algorithm and \ref{eq:potentially_tense}. By \Cref{def:pot_tense_vert} we have 
\begin{align*}
     \hat{d}_{\DSFront}^{t + b}(v_l) < \hat{d}^{\tHalf}_{\DSBack}(v) + d_\sigma(v, v_l) + (q - 2(j + 1))\epsilon \delta + r
\end{align*}
for all $v \in \suffix(\sigma, w_{next})$, since $w_{j + 1}'$ is the last vertex that fulfills inequality \ref{eq:potentially_tense}. Chaining this inequality with the guarantee given by \ref{eq:potentially_tense} for $w_j'$ yields
\begin{align*}
    \hat{d}^{\tHalf}_{\DSBack}(w_j') + d_\sigma(w_j', v_l) + (q - 2j)\epsilon \delta + r &< \hat{d}^{\tHalf}_{\DSBack}(v) + d_\sigma(v, v_l) + (q - 2j)\epsilon \delta + r - 2\epsilon \delta \\
    \iff \hat{d}^{\tHalf}_{\DSBack}(w_j') &< \hat{d}^{\tHalf}_{\DSBack}(v) - d_\sigma(w_j', v) - 2\epsilon\delta . \tag*{(error suffix)} \label{eq:error_suffix}
\end{align*}
But by assumption $w'_j$ is not tense, and thus there exists a vertex $v^\star$ after $w_j'$ on the path such that
\begin{align*}
    \hat{d}^{\tHalf}_{\DSBack}(v^\star) > \hat{d}^{\tHalf}_{\DSBack}(w_j') + 2\delta \tag*{(dist 1)} \label{eq:distOne}
\end{align*}
and 
\begin{align*}
    \hat{d}^{\tHalf}_{\DSBack}(v^\star) &\leq \hat{d}^{\tHalf}_{\DSBack}(w_j') + d_\sigma(w_j', v^\star) + \epsilon \delta \tag*{(dist 2)} \label{eq:distTwo}\\
    \iff \hat{d}^{\tHalf}_{\DSBack}(v^\star) - d_\sigma(w_j', v^\star) - \epsilon \delta &\leq \hat{d}^{\tHalf}_{\DSBack}(w_j') 
\end{align*}
by \Cref{def:tense_vertex}. We claim that we can conclude that $d_\sigma(w_j', v^\star) \leq d_\sigma(w_j', w_{j + 1}')$. This follows since if $v^\star \in \suffix(\sigma, w_{next})$, the last inequality combined with \ref{eq:error_suffix} where we set $v = v^\star$ would form a contradiction. Finally, we have 
\begin{align*}
    \hat{d}^{\tHalf}_{\DSBack}(w_j') + 2\delta &< \hat{d}^{\tHalf}_{\DSBack}(v^\star) \leq \hat{d}^{\tHalf}_{\DSBack}(w_j') + d_\sigma(w_j', v^\star) + \epsilon \delta \\
    \implies 2\delta - \epsilon \delta &< d_\sigma(w_j', v^\star) \leq  d_\sigma(w_j', w_{j + 1}')
\end{align*}
by chaining \ref{eq:distOne} with \ref{eq:distTwo} which concludes the proof for $\epsilon < 1$. 
\end{proof}

We combine our results in the following lemma.

\begin{lemma}
Let $r \in \mathbb{R}_{\geq 0}$ be a parameter and $\sigma = v_1, ..., v_l$ be a segment already present during the last global fixing phase at time $t$ such that
\begin{align*}
    slack^{\tHalf}_{\DSBack}(\sigma, \hat{d}^{t + b}_{\DSFront}(v_l)) - r = \mu \geq 4 \epsilon d_{\sigma}(v_1,v_l)
\end{align*}
where $b$ is the number of insertions that happened since the last global fixing phase during the processing of insertion $t$ and $0 \leq r \leq slack^{\tHalf}_{\DSBack}(\sigma, \hat{d}^{t + b}_{\DSFront}(v_l))$ is a parameter. Let $q = \floor{\mu/\epsilon \delta}$, then there is a subset of $\{w'_j\}$ of size at least $q/4 - 3$ consisting of vertices $w_j$ that are tense with regard to time $\tHalfInText$. Moreover, all such vertices $w_j$ satisfy:
\begin{enumerate}
    \item $|\hat{d}^{\tHalf}_{\DSBack}(w_i) - \hat{d}^{\tHalf}_{\DSBack}(w_j)| \geq \epsilon \delta$ for all $i \neq j$
    \item $\hat{d}^{\tHalf}_{\DSBack}(w_j) \leq \hat{d}^{t + b}_{\DSFront}(v_l) - r$
    \item $\hat{d}^{t+b}_{\DSBack}(v_l) \leq \hat{d}^{t + b}_{\DSFront}(v_l) - r$ if $w_j$ gets hit. 
\end{enumerate}
\label{lem:many_tense_vert}
\end{lemma}
\begin{proof}
We first show that at least $q/4 - 3$ of the vertices $w'_j$ are tense. By \Cref{lem:tens_vert_dist} our constructed potentially tense vertices are distinct, and thus we have $\floor{q/2}$ distinct potentially tense vertices. To arrive at a contradiction, we assume that $p \geq q/4 + 2$ of the potentially tense vertices $w_i$ are not actually tense. To use  \Cref{lem:not_tense_big_dist}, we focus on the first $\floor{q/2} - 1$ potentially tense vertices. These then must contain at least $q/4 + 1$ vertices that are not actually tense. From \Cref{lem:not_tense_big_dist} we conclude, that the segment has length $d(v_1, v_l) > q\delta/4 + \delta$, since there are $q/4 + 1$ distinct parts of length greater than $\delta$. But then 
\begin{align*}
    d(v_1, v_l) &> q\delta/4 + \delta \geq (\mu/\epsilon \delta - 1)\delta/4 + \delta \geq \mu/4\epsilon
\end{align*}
which is a contradiction to $4\epsilon d(v_1, v_l) \leq \mu$. Finally, let us prove the properties of tense vertices $w_i$:
\begin{enumerate}
    \item The first property directly follows from \Cref{lem:tens_vert_dist}.
    \item The second property is a direct consequence of condition \ref{eq:potentially_tense} in \Cref{def:pot_tense_vert}.
    \item For the third property, if any $w_i=w'_j$ gets hit, we have 
    \begin{align*}
        \hat{d}^{t+b}_{\DSBack}(v_l) \leq \hat{d}^{t}_{\DSBack}(v_l) \leq \hat{d}^{t}_{\DSBack}(w_j') + d_\sigma(w_j', v_l) \leq \hat{d}^{t+b}_{\DSFront}(v_l) - (q - 2j)\epsilon \delta - r \leq  \hat{d}^{t+b}_{\DSFront}(v_l) - r
    \end{align*}
where the first inequality follows from the fact that distance estimates only ever decrease in $\DSBack$, the second from \Cref{lem:hit}, the third from \ref{eq:potentially_tense} in \Cref{def:pot_tense_vert}, and the fourth from $j \leq \floor{q/2} - 1$. 
\end{enumerate}
\end{proof}

\paragraph{Tense segments}
Let's remind ourselfs that $\pi_{s,x}$ is a shortest $s-x$ path given by $s = v_1, ..., v_l = x$ at time $t + b$, that carries additive error $\hat{d}(x) - d^{t + b}(s,x) \geq 100 \epsilon \tau \lg n$. 

It is comprised of at most $b + 1$ maximal segments $\sigma_i = v_1^{(i)}, ...,  v_{l_i}^{(i)}$, that were already present during the last global fixing phase, as well as at most $b$ newly inserted edges $e_i = (v_{l_i}^{(i)},v_1^{(i + 1)})$ connecting them. A simple way to obtain the segments is to just remove all the edges that were inserted since time $t$ from the path, and look at what is left over. As one might expect, the index $i$ increases as we go along the path. Notice that some segments $\sigma_i$ could just contain a single vertex. We refer to the number of such segments with $p + 1 \leq B$. 
We define tense segments, which are to the path what tense vertices are to a segment.
\begin{definition}
For each segment $\sigma_i$ where $i \in [p]$, we define
\begin{align*}
    r_i = \argmin_{r \in \mathbb{R}}\left(\forall j > i, \forall v \in \sigma_j : \hat{d}^{t+b}_{\DSFront}(v) \geq \hat{d}^{t+b}_{\DSFront}(v_{l_i}^{(i)}) + d_{\pi_{s,x}}(v_{l_i}^{(i)}, v) - r + 5 (j - i) \epsilon \delta \lg n\right).
\end{align*}
For the segment $\sigma_{p+1}$ we define $r_{p+1} = 0$. If $r_i < slack_{\DSFront}^{t+b}(\sigma_i)$, we say the segment $\sigma_i$ is tense.
\label{def:tense_segment}
\end{definition}

\begin{lemma}
For all $i \in \{ 1, ..., p+1\}$, we have $r_i \geq 0$. 
\end{lemma}
\begin{proof}
The case $i = p+1$ follows by definition. For the other cases, assume the contrary $r_i < 0$ and consider the vertex $v_1^{(i+1)}$. By the definition of $r_i$, we have
\begin{align*}
    \hat{d}^{t+b}_{\DSFront}(v_1^{(i+1)}) &\geq \hat{d}^{t+b}_{\DSFront}(v_{l_i}^{(i)}) + d_{\pi_{s,x}}(v_{l_i}^{(i)}, v_1^{(i+1)}) - r_i + 5 \epsilon \delta \lg n
    \geq \hat{d}^{t+b}_{\DSFront}(v_{l_i}^{(i)}) + d_{\pi_{s,x}}(v_{l_i}^{(i)}, v_1^{(i+1)}) + 5 \epsilon \delta \lg n
\end{align*}
which is a contradiction to \Cref{inv:max_str_edge} since $d_{\pi_{s,x}}(v_{l_i}^{(i)}, v_1^{(i+1)}) = \omega(v_{l_i}^{(i)}, v_1^{(i+1)})$.
\end{proof}

We first show that a segment that moves a lot in distance estimate, cannot contribute more than $3\epsilon \delta \lg n$ additive error.

\begin{lemma}
    Let $\sigma_i = v_1^{(i)}, ..., v_{l_{i}}^{(i)}$ be a segment such that $\forall v \in \sigma_i : \hat{d}_{\DSBack}^{\tHalf}(v) \geq \hat{d}_{\DSBack}^{t+b}(v_1^{(i)}) + d_{\pi_{s,x}}(v_1^{(i)}, v) + 4\epsilon \delta\lg n$. Then $\hat{d}_{\DSBack}^{t+b}(v_{l_i}^{(i)}) \leq  \hat{d}_{\DSBack}^{t+b}(v_1^{(i)}) + d_{\pi_{s,x}}(v_1^{(i)},v_{l_i}^{(i)}) + 3\epsilon \delta \lg n$. 
    \label{lem:segment_pulled}
\end{lemma}
\begin{proof}
Assume, for the sake of contradiction, that $\hat{d}_{\DSBack}^{t+b}(v_{l_i}^{(i)}) > \hat{d}_{\DSBack}^{t+b}(v_1^{(i)}) + d_{\pi_{s,x}}(v_1^{(i)},v_l^{(i)}) + 3 \epsilon \delta \lg n$. Then, the segment has $slack^{t+b}_{\DSBack}(\sigma_i) \geq \hat{d}_{\DSBack}^{t+b}(v_{l_i}) - \hat{d}_{\DSBack}^{t+b}(v_1^{(i)}) - d_{\pi_{s,x}}(v_1^{(i)},v_{l_i}^{(i)}) > 3 \epsilon \delta \lg n$ and thus the assumptions of \Cref{lem:error_stays_t_tilde} are fulfilled. By  \Cref{lem:error_stays_t_tilde}, we have 
\begin{align*}
    slack^{\tHalf}_{\DSBack}(\sigma_i, \hat{d}_{\DSBack}^{t+b}(v_{l_i}^{(i)})) \geq \hat{d}_{\DSBack}^{t+b}(v_{l_i}^{(i)}) - \hat{d}_{\DSBack}^{t+b}(v_1^{(i)}) - d_{\pi_{s,x}}(v_1^{(i)},v_{l_i}^{(i)}) - 3 \epsilon \delta \lg n
\end{align*}
and therefore there exists a vertex $v_j^{(i)} \in \sigma_i$ witnessing this slack, so that
\begin{align*}
    \hat{d}_{\DSBack}^{t + b}(v_{l_i}^{(i)}) - \hat{d}_{\DSBack}^{\tHalf}(v_j^{(i)}) - d_{\pi_{s,x}}(v_j^{(i)},v_{l_i}^{(i)}) &\geq \hat{d}_{\DSBack}^{t+b}(v_{l_i}^{(i)}) - \hat{d}_{\DSBack}^{t+b}(v_1^{(i)}) - d_{\pi_{s,x}}(v_1^{(i)},v_{l_i}^{(i)}) - 3 \epsilon \delta \lg n. \\ \implies
    \hat{d}_{\DSBack}^{t+b}(v_1^{(i)}) &\geq \hat{d}_{\DSBack}^{\tHalf}(v_j^{(i)}) - d_{\pi_{s,x}}(v_1^{(i)},v_j^{(i)}) - 3 \epsilon \delta \lg n. \tag*{(witness)}\label{eq:witness}
\end{align*}
By assumption of the lemma, we have
\begin{align*}
    \hat{d}_{\DSBack}^{\tHalf}(v_j^{(i)}) \geq \hat{d}_{\DSBack}^{t+b}(v_1^{(i)}) + d_{\pi_{s,x}}(v_1^{(i)}, v_j^{(i)}) + 4\epsilon \delta\lg n. \tag*{(assumption)}\label{eq:precond}
\end{align*}
Chaining inequalities \ref{eq:witness} and \ref{eq:precond} yields 
\begin{align*}
    \hat{d}_{\DSBack}^{t+b}(v_1^{(i)}) &\geq \hat{d}_{\DSBack}^{t+b}(v_1^{(i)}) + \epsilon \delta \lg n \\
    0 &\geq \epsilon \delta \lg n 
\end{align*}
which is a contradiction. 
\end{proof}

Next we use the previously derived lemma, to show that if tense vertices in a tense segment get hit, all the vertices on the remainder of the path get reduced in distance estimate at some point. 

\begin{lemma}
Let $\sigma_i = v_1^{(i)}, ..., v_{l_{i}}^{(i)}$ be a tense segment so that 
\begin{align*}
    slack_{\DSFront}^{t+b}(\sigma) - r_i = \mu_i\geq 4d_{\pi_{s,x}}(v_1^{(i)}, v_{l_{i}}^{(i)}) + 3 \epsilon \delta \lg n.
\end{align*} 
Let $q = \mu_i/\epsilon\delta - 3\lg n$. Then the segment contained $q/4 - 3$ tense vertices $w_j$ in data structure $\DSBack$ at time $\tHalfInText$. For all such vertices $w_j$, we have further 
\begin{enumerate}
    \item $|\hat{d}^{\tHalf}_{\DSBack}(w_j) - \hat{d}^{\tHalf}_{\DSBack}(w_k)| \geq \epsilon \delta$ for all $j \neq k$
    \item $\hat{d}^{\tHalf}_{\DSBack}(w_j) \leq \hat{d}^{t + b}_{\DSFront}(v_{l_{i}}^{(i)}) - r_i$
    \item $\hat{d}^{t+b}_{\DSBack}(x) \leq \hat{d}^{t + b}_{\DSFront}(v_{l_{i}}^{(i)}) + d_{\pi_{s,x}}(v_{l_{i}}^{(i)}, x) - r_i + 4 \epsilon \tau \lg n$ if $w_j$ gets hit. 
\end{enumerate}
\label{lem:pullback}
\end{lemma}
\begin{proof}
We first calculate
\begin{align*}
    slack_{\DSBack}^{\tHalf}(\sigma, \hat{d}^{t+b}_{\DSBack}(v_{l_{i}}^{(i)})) - r_i &= slack_{\DSFront}^{t}(\sigma, \hat{d}^{t+b}_{\DSBack}(v_{l_{i}}^{(i)})) - r_i \\
    &\geq slack_{\DSFront}^{t+b}(\sigma, \hat{d}^{t+b}_{\DSBack}(v_{l_{i}}^{(i)})) - r_i - 2\epsilon \delta \lg n \\
    &= \mu_i  - 2\epsilon \delta \lg n
\end{align*}
where the first equality holds directly from the description of our algorithm,  the inequality holds because of \Cref{lem:error_stays}, and the last equality is just the definition of $\mu_i$. Now that we related the slack from data structure $\DSFront$ at time $t + b$ to slack in the data structure $\DSBack$ at time $\tHalfInText$, we can use \Cref{lem:many_tense_vert} and obtain a statement about tense vertices in $\DSBack$. We use the tense vertices $w_j$ as given by   \Cref{lem:many_tense_vert}. The first two points of this lemma directly follow directly from \Cref{lem:many_tense_vert}. Further, the third point of \Cref{lem:many_tense_vert} yields
\begin{align}\label{eq:base}
     \hat{d}^{t+b}_{\DSBack}(v_{l_i}^{(i)}) \leq \hat{d}^{t + b}_{\DSFront}(v_{l_i}^{(i)}) - r_i
\end{align}
if one of the tense vertices $w_j$ gets hit. We condition on $w_j$ being hit. To show that this implies the third point of our lemma, we set up an induction on index $k$ for value $i \leq k \leq p+1$ with induction hypothesis 
\begin{align*}
    \hat{d}^{t+b}_{\DSBack}(v_{l_k}^{(k)}) \leq \hat{d}^{t + b}_{\DSFront}(v_{l_{i}}^{(i)}) + d_{\pi_{s,x}}(v_{l_i}^{(i)}, v_{l_k}^{(k)}) - r_i + 4(k - i)\epsilon \delta \lg n.
\end{align*}
The base case is clear, since for $k = i$ this is just inequality \eqref{eq:base}. Now consider some step $i < k + 1$. Then
\begin{align*}
    \hat{d}^{t+b}_{\DSBack}(v_{1}^{(k+1)}) &\leq \hat{d}^{t+b}_{\DSBack}(v_{l_k}^{(k)}) + \omega(v_{l_k}^{(k)}, v_{1}^{(k+1)}) +  \epsilon\delta \\
    &\leq \hat{d}^{t + b}_{\DSFront}(v_{l_{i}}^{(i)}) + d_{\pi_{s,x}}(v_{l_i}^{(i)}, v_{1}^{(k+1)}) - r_i + 4(k - i)\epsilon \delta \lg n + \epsilon \delta \tag*{(step)}\label{eq:step} \\
    &\leq \hat{d}^{t + b}_{\DSFront}(v_{l_{i}}^{(i)}) + d_{\pi_{s,x}}(v_{l_i}^{(i)}, v_{1}^{(k+1)}) - r_i + 5(k - i)\epsilon \delta \lg n
\end{align*}
where the first inequality is due to \Cref{inv:max_str_edge}, and the second is an application of the induction hypothesis. By  \Cref{def:tense_segment} we have for all $v \in \sigma_{k+1}$ 
\begin{align*}
    \hat{d}^{t+b}_{\DSFront}(v) \geq \hat{d}^{t+b}_{\DSFront}(v_{l_i}^{(i)}) + d_{\pi_{s,x}}(v_{l_i}^{(i)}, v) - r_i + 5 (k + 1 - i) \epsilon \delta \lg n.
\end{align*}
Subtracting the previous two inequalities yields again for every $v \in \sigma_{k+1}$ 
\begin{align*}
    \hat{d}^{t+b}_{\DSFront}(v) - \hat{d}^{t+b}_{\DSBack}(v_{1}^{(k+1)}) &\geq d_{\pi_{s,x}}(v_{l_i}^{(i)}, v) - d_{\pi_{s,x}}(v_{l_i}^{(i)}, v_{1}^{(k+1)}) + 5 \epsilon \delta \lg n  \\
    \implies \hat{d}^{t+b}_{\DSFront}(v) &\geq \hat{d}^{t+b}_{\DSBack}(v_{1}^{(k+1)}) + d_{\pi_{s,x}}(v_{1}^{(k+1)}, v)  + 5 \epsilon \delta \lg n.
\end{align*}
Since $\hat{d}^{\tHalf}_{\DSBack}(v) = \hat{d}^{\tHalf}_{\DSFront}(v)\geq \hat{d}^{t+b}_{\DSFront}(v)$ because distance estimates only ever decrease, this means that we satisfy the assumptions of  \Cref{lem:segment_pulled}, and can thus conclude 
\begin{align*}
    \hat{d}_{\DSBack}^{t+b}(v_{l_{k+1}}^{(k+1)}) \leq  \hat{d}_{\DSBack}^{t+b}(v_1^{(k+1)}) + d_{\pi_{s,x}}(v_1^{(k+1)},v_{l_{k+1}}^{(k+1)}) + 3\epsilon \delta \lg n.
\end{align*}
By plugging in inequality \ref{eq:step} for $\hat{d}_{\DSBack}^{t+b}(v_1^{(k+1)})$ this yields
\begin{align*}
     \hat{d}_{\DSBack}^{t+b}(v_{l_{k+1}}^{(k+1)}) \leq \hat{d}^{t + b}_{\DSFront}(v_{l_{i}}^{(i)}) + d_{\pi_{s,x}}(v_{l_i}^{(i)}, v_{l_{k+1}}^{(k+1)})) - r_i + 4(k + 1 - i)\epsilon \delta \lg n
\end{align*}
which concludes our induction. Since $v_{l_{p+1}}^{(p+1)} = x$ and $p+1 \leq B \leq m^{1/3}$, we have 
\begin{align*}
    4(p + 1 - i)\epsilon \delta \lg n \leq 4\epsilon \tau \lg n
\end{align*} for $\delta = \tau/m^{1/3}$ and we conclude the proof.
\end{proof}

\paragraph{The slack of the path.} Finally, we address the very first point in our overview and analyse the total slack of some well-picked segments on the path $\pi_{s,x}$ at time $t+b$, which will allow us to argue about the number of tense vertices on $\pi_{s,x}$.  

To help define our sequence, we first define the forward segment of a segment.

\begin{definition}
Given a segment $\sigma_i$, we define the forward segment of $\sigma_i$ to be the segment $\sigma_j$ with $i < j$ for
\begin{align*}
    j = \min\left( j \in i + 1, ..., p + 1: \exists v \in \sigma_j : \hat{d}_{\DSFront}^{t+b}(v) = \hat{d}_{\DSFront}^{t+b}(v_{l_i}^{(i)}) + d_{\pi_{s,x}}(v_{l_i}^{(i)}, v) - r_i + 5 (j - i) \epsilon \delta \lg n\right).
\end{align*}
We define the forward-pair of $\sigma_i$ denoted by $forward(\sigma_i)$ to be the tuple $(\sigma_j, v_j^\star)$, where $v_j^\star \in \sigma_j$ is an arbitrary but fixed vertex such that 
\begin{align*}
    \hat{d}_{\DSFront}^{t+b}(v_j^\star) = \hat{d}_{\DSFront}^{t+b}(v_{l_i}^{(i)}) + d_{\pi_{s,x}}(v_{l_i}^{(i)}, v_j^\star) - r_i + 5 (j - i) \epsilon \delta \lg n. 
\end{align*}
\end{definition}

Next, we construct a (sub-)sequence of segments, using the previous definition. The construction of this sequence makes sure that there are many tense segments among them. 

\begin{definition}
We define a (sub-)sequence of segments $\sigma_{i(j)}$. We start with $\sigma_{i(1)} = \sigma_1$, and set $(\sigma_{i(j+1)}, v_{i(j+1)}^\star) = forward(\sigma_{i(j)})$ until $\sigma_{i(j)} = \sigma_{p+1}$. Let $p' + 1$ be the number of such segments. 
For convenience, we define $v_{i(1)}^\star$ to be some vertex $v$ in $\sigma_1$ that witnesses its slack, i.e. $slack^{t+b}_{\DSFront}(\sigma_1) = \hat{d}^{t+b}_{\DSFront}(v_{l_1}^{(1)}) - \hat{d}^{t+b}_{\DSFront}(v) - d_{\pi_{s,x}}(v, v_{l_1}^{(1)})$.
\label{def:forward_seq}
\end{definition}

We observe that $v_{i(j)}^\star$ is a vertex that witnesses the slack of segment $\sigma_{i(j)}$ by careful inspection of the definitions.

\begin{property}
For $j \in [p' + 1]$: $slack^{t+b}_{\DSFront}(\sigma_{i(j)}) = \hat{d}^{t+b}_{\DSFront}(v_l^{({i(j)})}) - \hat{d}^{t+b}_{\DSFront}(v_{i(j)}^\star) - d_{\pi_{s,x}}(v_{i(j)}^\star, v_l^{({i(j)})})$.
\label{lem:v_j_witnesses_slack}
\end{property}

Next we construct an upper bound for the distance estimate of the last vertex in a segment, which is important to show that the cumulative slack is high.
\begin{lemma}
For $2 \leq j \leq p' + 1$ we have 
\begin{align*}
    \hat{d}_{\DSFront}^{t+b}(v_{i(j)}^\star) \leq  d_{\pi_{s,x}}(s, v_{i(j)}^\star) + \sum_{k = 1}^{j - 1} slack^{t+b}_{\DSFront}(\sigma_{i(k)}) - r_{i(k)} + (i(j) - 1)5\epsilon \delta \lg n.
\end{align*}
\label{lem:sum_low_vert}
\end{lemma}
\begin{proof}
We show this claim by induction on $j$. For $j = 2$ we have 
\begin{align*}
    \hat{d}_{\DSFront}^{t+b}(v_{i(2)}^\star) &= \hat{d}_{\DSFront}^{t+b}(v_{l_1}^{(1)}) + d_{\pi_{s,x}}(v_{l_1}^{(1)}, v_{i(2)}^\star) - r_{i(1)} + 5 (i(2) - 1) \epsilon \delta \lg n \\
    &\leq slack_{\DSFront}^{t+b}(\sigma_{i(1)}) - r_{i(1)} + d_{\pi_{s,x}}(s, v_{i(2)}^\star) + 5 (i(2) - 1) \epsilon \delta \lg n. 
\end{align*}
where the first equality is by the defining property of $v_{i(2)}^\star$ and the second follows from the definition of slack, i.e. $slack_{\DSFront}^{t+b}(\sigma_{i(1)}) + d(s, v_{l_1}^{(1)}) \geq \hat{d}_{\DSFront}^{t+b}(v_{l_1}^{(1)})$.

For some $2 < j + 1 < p' + 1$ we similarly get
\begin{align*}
    \hat{d}_{\DSFront}^{t+b}(v_{i(j + 1)}^\star) &= \hat{d}_{\DSFront}^{t+b}(v_{l_{i(j)}}^{(i(j))}) + d_{\pi_{s,x}}(v_{l_{i(j)}}^{(i(j))}, v_{i(j + 1)}^\star) - r_{i(j)} + 5 (i(j + 1) - i(j)) \epsilon \delta \lg n \\
    &= \hat{d}_{\DSFront}^{t+b}(v_{i(j)}^\star) + slack_{\DSFront}^{t+b}(\sigma_{i(j)}) + d_{\pi_{s,x}}(v_{i(j)}^\star v_{i(j + 1)}^\star) - r_{i(j)} 5 (i(j + 1) - i(j)) \epsilon \delta \lg n \\
    &\leq d_{\pi_{s,x}}(s, v_{i(j)}^\star) + \sum_{k = 1}^{j - 1} slack^{t+b}_{\DSFront}(\sigma_{i(k)}) - r_{i(k)} + (i(j) - 1)5\epsilon \delta \lg n
\end{align*}
where the second equality is due to \Cref{lem:v_j_witnesses_slack} and the inequality is due to the induction hypothesis. This concludes our proof. 
\end{proof}

Finally, we show that the sum of the slack of our segments is large.

\begin{lemma}
$ \sum_{k = 1}^{p' + 1} slack^{t+b}_{\DSFront}(\sigma_{i(k)}) - r_{i(k)} \geq 95 \epsilon \tau \lg n$ 
\label{lem:sum_contributes}
\end{lemma}
\begin{proof}
We have \begin{align*}
    \hat{d}^{t+b}_{\DSFront}(x) = slack^{t+b}_{\DSFront}(\sigma_{i(p' + 1)}) + \hat{d}^{t+b}_{\DSFront}(v_{i(p' + 1)}^\star) + d_{\pi_{s,x}}(v_{i(p' + 1)}^\star, x)
\end{align*} by \Cref{lem:v_j_witnesses_slack}. If there is just one segment, we are done after subtracting $d(s,x)$ from both sides since $r_{p+1} = 0$. 

Otherwise, combining this with our previous \Cref{lem:sum_low_vert}, we obtain 
\begin{align*}
    \hat{d}^{t+b}_{\DSFront}(x) \leq  d_{\pi_{s,x}}(s, x) + \sum_{k = 1}^{p' + 1} slack^{t+b}_{\DSFront}(\sigma_{i(k)}) - r_{i(k)} + (i(p' + 1) - 1)5\epsilon \delta \lg n
\end{align*}
by direct calculation using $r_{i(p' + 1)} = r_{p + 1}  = 0$ by definition. Subtracting $d(s,x)$ from both sides we get
\begin{align*}
    100 \tau \lg n \leq \hat{d}^{t+b}_{\DSFront}(x) - d_{\pi_{s,x}}(s,x) &\leq  \sum_{k = 1}^{p' + 1} slack^{t+b}_{\DSFront}(\sigma_{i(k)}) - r_{i(k)} + (i(p' + 1) - 1)5\epsilon \delta \lg n \\
    &\leq \sum_{k = 1}^{p' + 1} slack^{t+b}_{\DSFront}(\sigma_{i(k)}) - r_{i(k)} + 5 \epsilon \tau \lg n
\end{align*}
where the second inequality follows from $i(p' + 1) \leq m^{1/3}$. We conclude our result by subtracting $5 \epsilon \tau \lg n$ on both sides.
\end{proof}

\begin{lemma}
For all $j \in [p' + 1]$
\begin{align*}
    \hat{d}^{t+b}_{\DSFront}(v_{l_{i(j)}}^{(i(j))}) - r_{i(j)} \leq  d_{\pi_{s,x}}(s, v_{l_{i(j)}}^{(i(j))}) + \sum_{k = 1}^{j} slack^{t+b}_{\DSFront}(\sigma_{i(k)}) - r_{i(k)} + 5 \epsilon \tau \lg n
\end{align*} 
\label{lem:dist_to_end}
\end{lemma}
\begin{proof}
For $j = 1$ this directly holds since 
\begin{align*}
    \hat{d}^{t+b}_{\DSFront}(v_{l_{i(1)}}^{(i(j))}) \leq d_{\pi_{s,x}}(s, v_{l_{i(1)}}^{(i(1))}) + slack^{t+b}_{\DSFront}(\sigma_{i(1)}) 
\end{align*}
by the definition of slack. For any other $j$, we have 
\begin{align*}
    \hat{d}_{\DSFront}^{t+b}(v_{i(j)}^\star) &\leq  d_{\pi_{s,x}}(s, v_{i(j)}^\star) + \sum_{k = 1}^{j - 1} slack^{t+b}_{\DSFront}(\sigma_{i(k)}) - r_{i(k)} + (i(j) - 1)5\epsilon \delta \lg n. \\
    &\leq  d_{\pi_{s,x}}(s, v_{i(j)}^\star) + \sum_{k = 1}^{j - 1} slack^{t+b}_{\DSFront}(\sigma_{i(k)}) - r_{i(k)} + 5\epsilon \tau \lg n
\end{align*}
by \Cref{lem:sum_low_vert} and $i(j) \leq m^{1/3}$. Combined with 
\begin{align*}
    slack^{t+b}_{\DSFront}(\sigma_{i(j)}) &= \hat{d}^{t+b}_{\DSFront}(v_l^{(i(j))}) - \hat{d}^{t+b}_{\DSFront}(v_{i(j)}^\star) - d_{\pi_{s,x}}(v_{i(j)}^\star, v_l^{{i(j)}}) \\
    \hat{d}^{t+b}_{\DSFront}(v_{i(j)}^\star) &= \hat{d}^{t+b}_{\DSFront}(v_l^{(i(j))}) -  d_{\pi_{s,x}}(v_{i(j)}^\star, v_l^{{i(j)}}) -   slack^{t+b}_{\DSFront}(\sigma_{i(j)}) 
\end{align*}
obtained from \Cref{lem:v_j_witnesses_slack} by reordering terms, we obtain our result by plugging in this derived value for $\hat{d}^{t+b}_{\DSFront}(v_{i(j)}^\star)$ and adding $d_{\pi_{s,x}}(v_{i(j)}^\star, v_l^{{i(j)}}) + slack^{t+b}_{\DSFront}(\sigma_{i(j)}) - r_{i(j)}$ on both sides. 
\end{proof}

\paragraph{The probability of failure.} We have now developed the necessary machinery to analyze the probability that our algorithm fails to maintain distance estimates correctly. We will show that it is negligible. To bound the probability of failure, we construct some tense segments that cumulatively contain a lot of tense vertices, and then show that if any of them was hit the potential would have decreased a lot. Finally, we argue that with high probability, a tense vertex is hit.

\begin{lemma}
Given the path $\pi_{s,x}$ as fixed throughout the section, we have that $\DSBack$ at time $\tHalfInText$ contains at least  $m^{1/3} \lg n$ tense vertices $w_j$ with respect to how we chose $\pi_{s,x}$. Further, we have for every $j \neq k$, that $|\hat{d}^{\tHalf}_{\DSBack}(w_j) - \hat{d}^{\tHalf}_{\DSBack}(w_k)| \geq \epsilon \delta$ and $\hat{d}^{t+b}_{\DSBack}(x) \leq \hat{d}^{t+b}_{\DSFront}(x) - \epsilon \tau$ if one of them gets hit between $\tHalfInText$ and $t$. 
\label{lem:tense_vertices_not_fixed}
\end{lemma}
\begin{proof}
We let 
\begin{align*}
    p^\star = \min\left( p \in {[p' + 1}] : \sum_{k = 1}^{p} slack^{t+b}_{\DSFront}(\sigma_{i(k)}) - r_{i(k)} \geq 60 \epsilon \tau \lg n\right).
\end{align*}
By \Cref{lem:sum_contributes} this is well defined. Next, we set 
\begin{align*}
    r'_{i(p^\star)} = r_{i(p^\star)} + \sum_{k = 1}^{p^\star} slack^{t+b}_{\DSFront}(\sigma_{i(k)}) - r_{i(k)} - 60 \epsilon \tau \lg n
\end{align*}
and $r'_{i(j)} = r_{i(j)}$ for $j \neq p^\star$. Then 
\begin{align*}
    \sum_{k = 1}^{p^\star} slack^{t+b}_{\DSFront}(\sigma_{i(k)}) - r'_{i(k)} = 60 \epsilon \tau \lg n
\end{align*}

Next, we make sure to fulfill the preconditions of \Cref{lem:pullback}. Therefore, we only look at elements of the sum $\sum_{i = k}^{p} slack^{t+b}_{\DSFront}(\sigma_{i(k)}) - r_{i(k)}$ so that $slack^{t+b}_{\DSFront}(\sigma_{i(k)}) - r_{i(k)} \geq 4\epsilon d(v_1^{(i(k))}, v_{l_{i(k))}}^{(i(k))}) + 3\epsilon \delta \lg n$. The summands that are not fulfilling this condition sum up to at most
\begin{align*}
    \sum_{k = 1}^{p' + 1} 4\epsilon d_{\pi_{s,x}}(v_1^{(i(k))}, v_{l_{i(k))}}^{(i(k))}) + 3\epsilon \delta \lg n
    &\leq \sum_{k = 1}^{p + 1} 4\epsilon d_{\pi_{s,x}}(v_1^{(k)}, v_{l_{(k)}}^{(k)}) + 3\epsilon \delta \lg n \\
    &\leq 4\epsilon d_{\pi_{s,x}}(s,x) + 3\epsilon \tau \lg n \\ &\leq 11 \epsilon \delta \lg n.
\end{align*}
Since the summands that are not fulfilling this condition sum up to at most $11 \epsilon \tau \lg n$, the remaining ones sum up to at least $49\epsilon \tau \lg n$. We apply \Cref{lem:pullback} to the all the remaining segments, and obtain that there are at least $m^{1/3}\lg n$ tense vertices in total among these segments, so that
\begin{align*}
    \hat{d}^{t+b}_{\DSBack}(x) \leq \hat{d}^{t + b}_{\DSFront}(v_{l_{i(j)}}^{(i(j))}) + d_{\pi_{s,x}}(v_{l_{i(j)}}^{(i(j))}, x) - r'_{i(j)} - 4 \epsilon \tau \lg n \tag*{(reduction)}\label{eq:reduction}
\end{align*}
if a tense vertex in segment $\sigma_{i(j)}$ got hit. Next, we show that \ref{eq:reduction} actually reduces the distance estimate of $t$ in $\DSBack$. By \Cref{lem:dist_to_end} we have
\begin{align*}
    \hat{d}^{t+b}_{\DSFront}(v_{l_{i(j)}}^{(i(j))}) - r_{i(j)} &\leq d_{\pi_{s,x}}(s, v_{l_{i(j)}}^{(i(j))}) + \sum_{k = 1}^{j} slack^{t+b}_{\DSFront}(\sigma_{i(k)}) - r_{i(k)} + 5 \epsilon \tau \lg n\\
    \hat{d}^{t+b}_{\DSFront}(v_{l_{i(j)}}^{(i(j))}) - r'_{i(j)} &\leq d_{\pi_{s,x}}(s, v_{l_{i(j)}}^{(i(j))}) + \sum_{k = 1}^{j} slack^{t+b}_{\DSFront}(\sigma_{i(k)}) - r'_{i(k)} + 5 \epsilon \tau \lg n \\
    &\leq d_{\pi_{s,x}}(s, v_{l_{i(j)}}^{(i(j))}) + 65 \epsilon \tau \lg n. \tag*{(level)}\label{eq:level}
\end{align*}
where the second step follows from the fact that $r'_{i(j)} = r_{i(j)}$ for $j \neq p^\star$.
Combining \ref{eq:reduction} with \ref{eq:level} yields
\begin{align*}
    \hat{d}^{t+b}_{\DSBack}(x) \leq d_{\pi_{s,x}}(s,x) + 70 \epsilon \tau \lg n
\end{align*}
which means that if any of the tense vertices in these segments got hit, we have $\hat{d}^{t+b}_{\DSBack}(x) \leq \hat{d}^{t+b}_{\DSFront}(x) - \epsilon \tau$ since $\hat{d}^{t+b}_{\DSBack}(x) \geq d_{\pi_{s,x}}(s,x) + 100\epsilon \tau \lg n$. 

The tense vertices in the same segment have $\epsilon \delta$ pairwise difference in distance estimates and the tense vertices of segment $\sigma_i(j)$ have distance estimate less than $\hat{d}^{t+b}_{\DSFront}(v_{l_{i(j)}}^{(i(j)}) - r'_{i(j)}$ by  \Cref{lem:pullback}. But then for $i(j) < i(k)$ we have that all the vertices in $v \in \sigma_{i(k)}$ have distance estimate 
\begin{align*}
    \hat{d}^{t+b}_{\DSFront}(v) &\geq \hat{d}^{t+b}_{\DSFront}(v_{l_{i(j)}}^{({i(j)})}) + d_{\pi_{s,x}}(v_{l_{i(j)}}^{(i)}, v) - r_{i(j)} + 5 (i(k) - {i(j)}) \epsilon \delta \lg n \\
    &\geq \hat{d}^{t+b}_{\DSFront}(v_{l_{i(j)}}^{({i(j)})})  - r'_{i(j)} + 5 \epsilon \delta  \lg n
\end{align*}
by \Cref{def:tense_segment} and since $r'_{i(j)} = r_{i(j)}$ for $j \neq p^\star$. This shows that they are sufficiently far apart and concludes our proof. 
\end{proof}

\begin{lemma}
Consider some vertex $x$ with $\hat{d}^{t + b}_{\DSFront}(x) - \hat{d}^{t + b}_{\DSBack}(x) \geq \epsilon \tau$ at time $t + b$. Then the potential $\phi_{\DSBack} = \sum_{v \in V}  \hat{d}_{\DSBack}^{t+b}(v)$ decreased by at least $\frac{1}{4} \epsilon m^{1/3} \tau$ since time $\tHalfInText$.
\label{lem:pot_dec}
\end{lemma}
\begin{proof}
Note that $\hat{d}^{t + b}_{\DSFront}(v)$ is a lower bound for $\hat{d}^{t}_{\DSFront}(v)$, and thus for $\hat{d}^{\tHalf}_{\DSBack}(v) = \hat{d}^{t}_{\DSFront}(v)$ as well. Consider the vertex $u$ that caused the final decrease of $v$. We must have $\hat{d}^{t_{last} + t}_{\tilde{G}}(u) \leq \hat{d}^{t_{last} + t}_{\tilde{G}}(v) - \omega(u,v)$ and as well as 
\begin{align*}
    \hat{d}^{t+b}_{\DSBack}(u) \geq \hat{d}^{t+b}_{\tilde{G}}(v) - \omega(u,v) - \epsilon \delta
\end{align*}
and
\begin{align*}
       \hat{d}^{t + b}_{\DSFront}(u) \geq \hat{d}^{t+b}_{\DSFront}(v) - \omega(u,v) - \epsilon \delta
\end{align*}
Thus, we have 
\begin{align*}
    \hat{d}^{\tHalf}_{\DSBack}(u) = \hat{d}^{t}_{\DSFront}(u) &\geq \hat{d}^{t+b}_{\DSFront}(v) - \omega(u,v) - \epsilon \delta\\
    &\geq \hat{d}^{t + b}_{\DSBack}(v) + \epsilon \tau - \omega(u,v) - \epsilon \delta  \\
    &\geq \hat{d}^{t + b}_{\DSBack}(u) + \epsilon \tau - \epsilon \delta
\end{align*}
Therefore we can iterate the argument $\tau/2\delta = m^{1/3}/2$ times and obtain that all these vertices decreased by at least $\epsilon \tau/2$. They are clearly source vertices of at least one edge each. Therefore the potential decrease is at least $\epsilon m^{1/3} \tau/4$ as desired. 
\end{proof}

Having shown the previous lemma, it just remains to prove it is unlikely none of all these tense vertices got hit. We conclude with the proof of \Cref{lma:keyLemmaPotentialReduction}.

\begin{proof}[Proof of \Cref{lma:keyLemmaPotentialReduction}]
By \Cref{lem:tense_vertices_not_fixed} there are at least $m^{1/3}\lg n$ vertices $w_i$ tense vertices, such that $|\hat{d}^{\tHalf}(w_i) - \hat{d}^{\tHalf}(w_j)| \geq \epsilon \delta$ for $i \neq j$. Therefore, there were at least $\epsilon m^{1/3} \lg n$ distinct distance estimate ranges $[r_k, r_k + \delta)$ of width $\delta$ containing tense vertices. The probability that we sample a single $i$ from
\begin{align*}
    0, ..., \ceil{2m^{1/3} + 200\epsilon m^{1/3} \lg n - 8}
\end{align*} so that  $[r_k, r_k + 3\delta] \subseteq [i, i + \delta)$ for some $k$, and thus the tense vertices in $[r_k, r_k + \delta]$ get hit is at least $\epsilon/200$. After $1/\epsilon$ independent sampling steps, the expected number of sampled tense vertices is at least $1/200$. By Markov's inequality, the probability of not having sampled a single one of them is at most $199/200$. We repeat this procedure $2000 \log n$ times, and bound the probability of never having sampled such an $i$ with $1/n^{5}$ using a standard Chernoff bound argument. If such a tense vertex is sampled, we have a decrease of $\frac{1}{4} \epsilon m^{1/3} \tau$ in potential by \Cref{lem:pot_dec}.

Finally, we notice that our proof only uses the randomness introduced by the last global fixing phase, which is not revealed to the adversary. Therefore, an adaptive adversary has no advantage over an oblivious one.
\end{proof}

\newpage

\bibliographystyle{alpha}
\bibliography{refs}

\newpage

\appendix

\section{Related Work} \label{subsec:relWork}

Here we also give an overview of more broadly related work. We use the same notational conventions as in the \Cref{subsec:priorWork}.

\paragraph{Directed APSP.} For the partially-dynamic All-Pairs Shortest Paths problem, a near-optimal algorithm with approximation $(1+\epsilon)$ and running time $\tilde{O}(mn \polylog(W))$ was given by Bernstein \cite{bernstein2016maintaining} that works against an oblivious adversary. Recent work in this setting has therefore focused on removing the oblivious adversary assumption \cite{karczmarz2019reliable, evald2020decremental} both of which currently achieve total update time $\tilde{O}(n^{2+1/3} \polylog(W))$ in sparse graphs where $m = \tilde{O}(n)$. For dense graphs, a deterministic $\tilde{O}(n^3 \polylog(W))$ update time algorithm exists \cite{karczmarz2020simple}.

There is also significant research on the All-Pairs Shortest Paths problem in the fully dynamic setting. Most notably, an algorithm by Italiano and Demetrescu \cite{demetrescu2004new} that achieves amortized update time $\tilde{O}(n^2)$ per update which was thereafter improved in logarithmic factors, simplified and generalized by Thorup \cite{thorup2004fully}. There has also been significant research on fully dynamic APSP with worst-case update times \cite{thorup2005worst, abraham2017fully, van2019dynamic, gutenberg2020fully, bergamaschi2021new}.

\paragraph{Undirected SSSP.} The undirected partially-dynamic version of the SSSP problem, in a major breakthrough, was solved by Henzinger, Krinninger and Nanongkai \cite{HKN14} who gave the first near-optimal $m^{1+o(1)}\polylog(W)$ time $(1+\epsilon)$ approximation algorithm although only against an oblivious adversary. Recent efforts to derandomize/ strengthen the adversary model \cite{bernstein2016deterministic, bernstein2017deterministic, bernstein2017deterministicWeighted,chuzhoy2019new,gutenberg2020deterministic, bernstein2020fully, chuzhoy2021deterministic} has been an extensive branch of research and has recently culminated in a deterministic algorithm \cite{bernstein2021deterministic} that also achieves total update time $m^{1+o(1)}\polylog(W)$. We point out that while some of the algorithms mentioned in this paragraph only work in decremental graphs, we believe that both the results \cite{HKN14} and \cite{bernstein2021deterministic} extend rather seamlessly to the incremental setting.

\paragraph{Undirected APSP.} Finally, we remark that there is also extensive literature about the partially-dynamic APSP problem in undirected graphs. For $(1+\epsilon)$-approximation, a near-optimal deterministic algorithm follows from  \cite{bernstein2021deterministic} with total update time $mn^{1+o(1)} \polylog(W)$. In the decremental setting, there are also results for larger approximations: a recent result by Chechik \cite{chechik2018near} achieves near-optimal update time $mn^{1/k + o(1)}\polylog(W)$ and reports distance with a $(1+\epsilon)(2k-1)$ stretch, the algorithm assumes an oblivious adversary. Against an adaptive adversary, two recent results \cite{chuzhoy2021decremental, bernstein2021deterministic} obtain $m^{1+o(1)}\polylog(W)$ update time for reporting $n^{o(1)}$-approximate shortest path distances.

\section{The Hop Set Barrier and the ES-tree}  \label{subsec:hop_set_barrier}

In this section we explain the usage of the ES-tree \cite{SE81, HK95} in conjunction with a hop set for the unfamiliar reader. Although we present them in the incremental setting, these techniques naturally extend to the decremental one.  

For a graph with integer edge weights up to $W$, the ES-tree data structure initially stores the distance of each vertex to a dedicated source $s$. Then, whenever an edge $(u,v)$ is inserted, it checks if $v$ profits from using this edge. If so, $v$ gets decreased and such checks are run recursively for all out-edges of decreased vertices until nobody profits anymore. It is easy to see that this technique maintains distances to the source exactly. The total update time is $O(mnW)$ because the tail of each edge starts out at some distance upper bounded by $(n-1)W$, and the edge is explored whenever this distance decreases by at least one. 

A hop set with hop $h$ is a collection of short cut edges that guarantee the existence of a shortest path using at most $h$ edges between any two vertices $u$ and $v$. Since such a set implies maximum distance $hW$, the run-time immediately improves to $O(mhW)$ by adding the short cut edges to the graph.

Finally, standard edge rounding techniques yield a $(1 + \epsilon)$-approximation scheme with total update time $\tilde{O}(mh \log W/\epsilon)$ \cite{bernstein2011improved}. To obtain this, we distinguish between path length ranges $[\tau, 2\tau)$ as in the algorithms presented in this article. For a particular range, all edge weights $\omega(e)$ are rounded to $\tilde{\omega}(e) = \ceil{h\omega(e)/\epsilon\tau}$. We set the maximum maintained distance estimate to $3\ceil{h/\epsilon}$ and use a standard ES-tree as discussed above. For a $s-x$ path $v_1, ..., v_l$ with $l \leq h$ in the distance range $d(s,x) \in [\tau, 2\tau)$, we have
\begin{align*}
    d(s, x) = \sum_{i = 1}^{l - 1} \omega(v_i, v_{i + 1}) \leq \sum_{i = 1}^{l - 1} (\epsilon\tau /h) \tilde{\omega}(v_i, v_{i + 1}) \leq \sum_{i = 1}^{l - 1} \omega(v_i, v_{i + 1}) + \epsilon\tau /h \leq d(s,x) + \epsilon \tau.
\end{align*}
Thus, we obtain a $(1+\epsilon)$-approximation by maintaining separate data structures for all possible path length ranges $[2^i, 2^{i+1})$ and reporting the minimum among the distance estimates after scaling them back with $\tau \epsilon/h$.

Conditional on the barrier $h = \tilde{\Omega}(\sqrt{n})$, this technique cannot achieve total update time $O(n^{3/2 - c})$ for any small positive constant $c$, even for very sparse graphs. This is referred to as the hop set barrier.
\end{document}